\documentclass[lettersize,journal]{IEEEtran}
\usepackage{amsmath,amsfonts,amsthm,amssymb,bbm}
\usepackage{algorithmicx}
\usepackage{algorithm}
\usepackage{algpseudocode}
\usepackage{array}
\usepackage[caption=false,font=footnotesize,labelfont=rm,textfont=rm]{subfig}
\usepackage{textcomp}
\usepackage{stfloats}
\usepackage{url}
\usepackage{verbatim}
\usepackage{graphicx}
\usepackage{epstopdf} % 自动转换 EPS 到 PDF
\usepackage{cite}
\usepackage[colorlinks,linkcolor=blue]{hyperref}
\usepackage{bm}
\usepackage{gensymb}
\usepackage{booktabs}

\usepackage{color}

\newtheorem{theorem}{\textbf{\textit{Theorem}}}

\newtheorem{assumption}{\textit{Assumption}}
\graphicspath{{figure/}{/bios/}}

\begin{document}

\iffalse	
\title{Joint Optimization of Latency and Accuracy for Split Federated Learning in User-Centric Cell-Free MIMO Networks}
\fi

\title{User-Centric Cell-Free MIMO based Split Federated Learning: Joint Optimization of Latency and Accuracy}

\author{ Zitong Wang,~\IEEEmembership{Student Member,~IEEE,}  Cheng Zhang,~\IEEEmembership{Member,~IEEE,}  Wen Wang,~\IEEEmembership{Member,~IEEE,} Shuigen Yang, Haiming Wang, Yongming Huang,~\IEEEmembership{Fellow,~IEEE}
\thanks{Z.~Wang, C.~Zhang, and Y.~Huang are with the National Mobile Communication Research Laboratory, Southeast University, Nanjing 210096, China, and also with the Purple Mountain Laboratories, Nanjing 211111, China (e-mail: ztwang2410@seu.edu.cn; zhangcheng\_seu@seu.edu.cn; huangym@seu.edu.cn). \it{(Corresponding authors: Yongming Huang.)}}
\thanks{W.~Wang is with the Purple Mountain Laboratories, Nanjing 211111, China (e-mail:  wangwen@pmlabs.com.cn).}
\thanks{S.~Yang is with the Lenovo Group, Lenovo Research, Shanghai 201203, China (e-mail:yangsg3@lenovo.com).}
\thanks{H.~Wang is with the Lenovo Group, Lenovo Research, Beijing 2100094, China (e-mail:wanghm14@lenovo.com).}
}

\markboth{IEEE Transactions on WIRELESS COMMUNICATIONS}%
{Submitted paper}
\maketitle

\begin{abstract}

This paper proposes a user-centric split federated learning (UCSFL) framework for user-centric cell-free multiple-input multiple-output (CF-MIMO) networks to support split federated learning (SFL). In the proposed UCSFL framework, users deploy split sub-models locally, while complete models are maintained and updated at access point (AP)–side distributed processing units (DPUs), followed by a two-level aggregation procedure across DPUs and the central processing unit (CPU). Under standard machine learning (ML) assumptions, we provide a theoretical convergence analysis for UCSFL, which reveals that the AP-cluster size is a key factor influencing model training accuracy. Motivated by this result, we introduce a new performance metric, termed the latency-to-accuracy ratio, defined as the ratio of a user’s per-iteration training latency to the weighted size of its AP cluster. Based on this metric, we formulate a joint optimization problem to minimize the maximum latency-to-accuracy ratio by jointly optimizing uplink power control, downlink beamforming, model splitting, and AP clustering. The resulting problem is decomposed into two sub-problems operating on different time scales, for which dedicated algorithms are developed to handle the short-term and long-term optimizations, respectively. Simulation results verify the convergence of the proposed algorithms and demonstrate that UCSFL effectively reduces the latency-to-accuracy ratio of the VGG16 model compared with baseline schemes. Moreover, the proposed framework adaptively adjusts splitting and clustering strategies in response to varying communication and computation resources. An MNIST-based handwritten digit classification example further shows that UCSFL significantly accelerates the convergence of the VGG16 model.

\end{abstract}


\IEEEpeerreviewmaketitle

\section{Introduction}

Sixth-generation (6G) networks are anticipated to achieve transformative advancements, fundamentally driven by the deep integration with artificial intelligence (AI) \cite{you2025ai}. 
A key manifestation of this integration is “AI at the edge,” which enhances service capabilities locally, thereby providing higher reliability and lower latency for AI service assurance.
However, this integration imposes unprecedented demands on both computation and communication, resulting in severe latency, energy consumption, network congestion, and privacy risks across both training and inference phases  \cite{9606720}. 
To mitigate these challenges at the network level, several foundational 6G technologies have been envisioned, among which cell-free multiple-input multiple-output (CF-MIMO) stands out as a promising architectural paradigm to improve the performance of wireless networks by eliminating the “cells” of traditional mobile communications \cite{wang2023full}.
In particular, user-centric CF-MIMO networks enhance system performance through macro diversity, power efficiency, advanced interference management, and robust connectivity, thereby ensuring uniform coverage and consistent quality of service for all users \cite{9650567}.
Complementarily, at the algorithmic level, distributed machine learning (DML) has emerged as a key paradigm that enables edge devices to collaboratively train models without exchanging raw data. 
By localizing data processing, DML significantly reduces communication overhead and training latency while simultaneously enhancing data privacy \cite{9446488, 9562559, 9733984}.

As an important realization of DML, federated learning (FL) has attracted considerable attention in wireless networks and is widely regarded as a key enabler for ubiquitous AI in 6G networks \cite{9205981}. 
It provides a promising framework for training AI models directly on wireless devices while maintaining service quality.
Nevertheless, the frequent transmission of full model parameters in conventional FL often incurs substantial communication latency, and the heterogeneity of user equipments (UEs) further exacerbates limitations in computational and storage resources  \cite{9460016, 9141214, 9084352}.
To cope with these limitations, split federated learning (SFL) has emerged as an attractive alternative by integrating model splitting into the FL training pipeline.
\footnote{Here, SFL broadly describes frameworks unifying model splitting and federated aggregation, thereby covering methodologies like federated split learning and hybrid split-and-federate learning under its conceptual scope.}
In particular, the SFL technique does not share the entire model, and only a designated sub-model is deployed at each UE, while the remaining sub-model is maintained at the server side.
This design effectively reduces communication overhead and computational demand at the UE, thereby improving energy efficiency and enhancing model privacy \cite{hafi2024split}.

Recently, many research efforts have been devoted to the design of SFL frameworks under heterogeneous UE capabilities and diverse channel conditions.
Each of the works in \cite{10251444, 9923620, 10522472} proposed its own novel SFL algorithm to to enable efficient distributed training while explicitly accounting for the limited computational resources of heterogeneous UEs.
In \cite{11091500}, a communication- and storage-efficient federated split learning scheme was developed, where an auxiliary network was employed to locally update UE weights while maintaining a single global model at the server.
Furthermore, the authors of \cite{10547401} proposed a dynamic federated split learning framework to address the joint problem of data and resource heterogeneity in IoT distributed training, achieved by introducing resource-aware split computing and dynamic clustering mechanisms.

Meanwhile, another major line of research has focused on improving key performance metrics of SFL, including latency, communication/computation cost, and privacy \cite{10304624, 11299451, 10382156, 11119144, 10587192, 11130508}.
In \cite{10304624}, an SFL framework was proposed to minimize system latency under UE resource constraints.
The work in \cite{11299451} introduced a vehicular edge hybrid federated split learning paradigm to reduce per-round training latency while accounting for privacy protection and resource limitations in the Internet of Vehicles. 
A communication-efficient FL scheme with split-layer aggregation, termed FedSL, was proposed in \cite{10382156} to reduce communication overhead by transmitting fewer model parameters. 
In \cite{11119144}, a communication- and computation-efficient SFL framework was developed, which enabled dynamic model splitting and the broadcasting of aggregated gradients. 
The authors of \cite{10587192} proposed an efficient split learning approach within the FL framework that jointly reduced computation and communication costs while maintaining high model accuracy. 
More recently, the authors of \cite{11130508} presented a privacy-aware SFL scheme that addressed the accuracy–efficiency–privacy trade-off in large language model fine-tuning over heterogeneous Internet of Things (IoT) devices.

However, few studies have focused on deploying SFL frameworks over CF-MIMO networks. 
To inform this direction, we review key studies on implementing FL within CF-MIMO networks, whose insights provide a foundation for addressing SFL deployment challenges.
In \cite{11196025}, the authors studied an FL realization in a CF-mMIMO system with multi-antenna access points (APs) and UEs, employing multi-datastream transmission to harness the inherent benefits of multi-antenna transmission.
The work in \cite{11235610} introduced a device-centric cell-free network to mitigate the negative effects of random fading and limited radio resources on FL, and evaluated the impact of communication and computation factors on FL performance.
Besides, the work in \cite{10513357} conducted an theoretical analysis to examine the impact of wireless transmission errors on FL process within the framework of CF-MIMO networks. 

Beyond architectural deployment, considerable attention has also been paid to optimizing key performance metrics, such as latency, accuracy, and energy consumption, for FL deployed over CF-MIMO networks \cite{9124715, 9796621, 10943191, 10937196, 11126942, 10120750}.
In \cite{9124715}, a general CF-MIMO-based framework was proposed to support FL with the objective of minimizing training time.
The authors of \cite{9796621} investigated user selection strategies in CF-MIMO systems to minimize total FL execution time by selecting users with favorable channel conditions.
The work in \cite{10943191} studied the joint optimization of AP clustering and power allocation for CF-MIMO-enabled FL training time minimization, taking into account clients’ mobility at walking speed and heterogeneous computing capabilities. 
In \cite{10937196}, uplink power allocation was optimized to maximize the number of global FL iterations while jointly considering uplink energy and latency.
The authors of \cite{11126942} proposed a joint communication and computation scheme for CF-MIMO networks to maximize the FL iteration error gap within fixed time limit.
In \cite{10120750}, the authors designed an energy-efficient FL scheme over cell-free IoT networks by minimizing the total energy consumption of participating devices.

A review of these existing works reveals two dominant research thrusts. 
The first focuses on architectural innovations within SFL itself to enhance its training performance.
The second leverages advanced 6G communication technologies, particularly CF-MIMO networks, as a platform to deploy and optimize FL.
However, a critical gap remains at their intersection: few studies have investigated the deployment and optimization of SFL frameworks within user-centric CF-MIMO networks. 
Moreover, the joint optimization of training latency and accuracy, which inherently exhibit a trade-off relationship in FL/SFL systems, remains largely under-explored. 
This gap motivates our work.

In this paper, we propose a novel user-centric split federated learning (UCSFL) framework over a user-centric CF-MIMO network to jointly optimize training latency and accuracy.
By leveraging AP-side distributed processing units (DPUs), the proposed framework reduces both transmission and computation latency, while improving model aggregation accuracy through adaptive model splitting and AP clustering strategies, thereby accelerating the convergence process.
The main contributions are summarized as follows:

\begin{itemize}
	\item We propose the UCSFL framework, which enables UEs to split their models and deploy sub-models locally, while complete models are maintained and updated at AP-side DPUs. 
	This architecture facilitates a two-level model aggregation procedure across DPUs and the CPU.
	Based on common ML assumptions, we provide a theoretical convergence analysis, which demonstrates that each UE's AP-cluster size is a key factor in determining its model training accuracy.	
	
	\item We define a new performance metric for each UE, termed the latency-to-accuracy ratio, which captures the trade-off between training latency and model accuracy. 
	Specifically, it is defined as the ratio of the UE's per-iteration training latency to the weighted size of its AP-cluster.
	On this basis, we formulate a joint resource allocation problem for uplink power control, downlink beamforming, model splitting, and AP clustering, with the goal of minimizing the maximum latency-to-accuracy ratio among all UEs.
	
	\item We decompose the joint optimization problem of latency-accuracy  into two sub-problems operating on different time scales, which are addressed by two dedicated algorithms, respectively. 
	For the non-convex short-term sub-problem, we develop an efficient nested block coordinate descent (NBCD) algorithm to iteratively optimize the power allocation and beamforming. 
	For the long-term sub-problem involving discrete decisions, we propose a multi-agent intelligent search (MAIS) algorithm within a deep reinforcement learning (DRL) framework to optimize model splitting and AP clustering strategies.

	\item Simulation results verify the convergence of the proposed algorithms and show that the proposed UCSFL framework significantly reduces the latency-to-accuracy ratio of VGG16 model compared with other baseline schemes, thereby improving training performance and accelerating SFL convergence.
	In addition, experiments on handwritten digit classification show that UCSFL effectively accelerates the model convergence while maintaining the training accuracy.
\end{itemize}

The rest of this article is organized as follows. 
Section \ref{System Model} presents a user-centric CF-MIMO system model and the UCSFL framework, along with its convergence analysis.
Section \ref{Problem Formulation} formulates the optimization problem of minimizing the latency-to-accuracy ratio, while Section \ref{Proposed Algorithm} proposes two algorithms to solve the short-term sub-problem and long-term sub-problem, respectively. 
Section \ref{Simulation} verifies the performance of the proposed UCSFL through numerical examples. 
Finally, Section \ref{Conclusion} concludes this article.

\section{System Model}
\label{System Model}
	\begin{figure}[tp]
		\centering
		\includegraphics[scale=0.6]{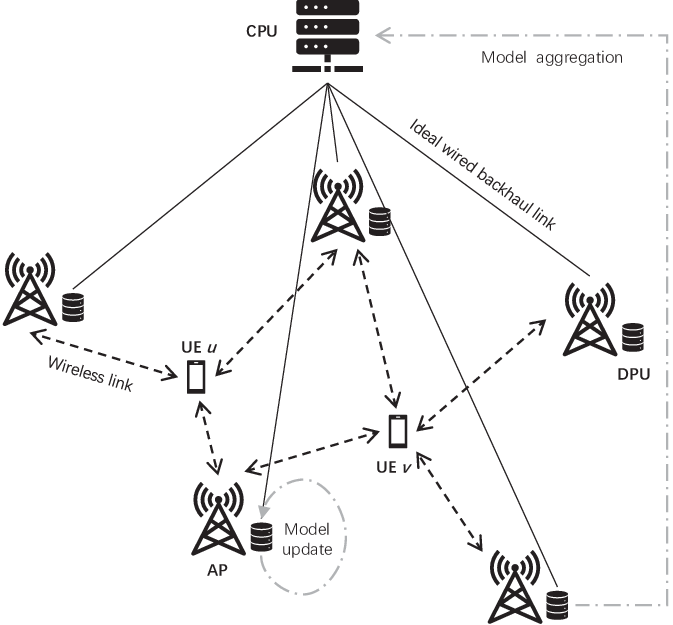}
		\caption{A user-centric CF-MIMO network to support FL.}
		\label{2_1}
	\end{figure}
	We consider a user-centric CF-MIMO network consisting of $U$ UEs and $M$ APs connected to a single central processing unit (CPU) via ideal wired backhaul links, as shown in Fig. \ref{2_1}.
	Each AP $m\in\mathcal{M}=\{1,2,3,...,M\}$ is equipped with $A$ antennas, while each UE $u\in\mathcal{U}=\{1,2,3,...,U\}$ is single antenna.
	Besides, each AP $m$ is equipped with a distributed processing unit (DPU) with a certain amount of computing power, which is represented by the computing frequency $f^{dpu}$.
	Similarly, the computing power of UE can be represented by $f^{ue}$.
	Without loss of generality, we assume $f^{ue}<f^{dpu}$.

	Meanwhile, we assume that only partial APs can be associated with UE $u$ depending on the demand of user-centric. 
	The association between UE $u$ and AP $m$ is represented by a binary variable $b_{m,u}\in\left\{0,1\right\}$. 
	Specifically, $b_{m,u}=1$ indicates that AP $m$ serves UE $u$, and $b_{m,u}=0$ otherwise.
	 
\subsection{Wireless Transmission in User-Centric Cell-Free MIMO System}
	
	The received signal of AP $m$ in the uplink time slot $t$ is expressed as
	\begin{equation}
		\mathbf{y}_{m,t}^{ul}=\mathbf{H}_{m}\mathbf{P}^{ul}\mathbf{s}_{t}^{ul}+\mathbf{z}_{m}^{ul},
		\label{1}
	\end{equation}
	where $\mathbf{y}_{m,t}^{ul}\in\mathbb{C}^{L\times{}1}$ is the superimposed signal vector from all UEs received by AP $m$, $\mathbf{H}_{m}=\left[\mathbf{h}_{m,1}^{\mathsf{T}},...,\mathbf{h}_{m,U}^{\mathsf{T}}\right]\in\mathbb{C}^{L\times{}U}$ represents the channel matrix from AP $m$ to all UEs, $\mathbf{P}^{ul}=diag\left(\sqrt{p_{1}^{ul}},...,\sqrt{p_{U}^{ul}}\right)\in\mathbb{R}^{U\times{}U}$ represents the UE-side power coefficient matrix, $\mathbf{s}_{t}^{ul}=\left[s_{1,t}^{ul},...,s_{U,t}^{ul}\right]\in\mathbb{R}^{U\times{}1}$ represents the data vector sent by UE side in the uplink time slot $t$, and $\mathbf{z}_{m}^{ul}\sim\mathcal{CN}\left(0,\sigma^2\mathbf{I}_{L\times{}1}\right)$ represents the channel Gaussian white noise vector with power $\sigma^2$.
	
	We define the combiner beamforming matrix as $\mathbf{U}_{m}=\left[\mathbf{u}_{m,1}^{\mathsf{T}},...,\mathbf{u}_{m,U}^{\mathsf{T}}\right]^{\mathsf{H}}\in\mathbb{C}^{U\times{}L}$, where $\mathbf{u}_{m,u}=\frac{\mathbf{h}_{m,u}}{\|\mathbf{h}_{m,u}\|}$ represents the conjugate beamforming.
	We multiply the combiner matrix with the received signal matrix to obtain the estimation of $U$ UEs' symbols at AP $m$ in the uplink time slot $t$, that is
	\begin{equation}
		\overline{\mathbf{y}}_{m,t}^{ul}=\mathbf{U}_{m}\mathbf{H}_{m}\mathbf{P}^{ul}\mathbf{s}_{t}^{ul}+\mathbf{U}_{m}\mathbf{z}_{m}^{ul},
		\label{2}
	\end{equation}
	where $\overline{\mathbf{y}}_{m,t}^{ul}\in\mathbb{R}^{U\times{}1}$ is the combined signal vector and each element in the vector represents the estimated symbol of corresponding UE $u$. We superpose the estimated signal received from all APs at the CPU and obtain that
	\begin{equation}
		\begin{aligned}
		\widetilde{s}_{u,t}^{ul}&=\sum\limits_{m\in\mathcal{M}}{b_{m,u}\overline{y}_{m,u,t}^{ul}}\\
		&=\sum\limits_{m\in\mathcal{M}}{\sqrt{p_{u}^{ul}}b_{m,u}\mathbf{h}_{m,u}{\mathbf{u}_{m,u}^{\mathsf{H}}}s_{u,t}^{ul}}\\
		&\quad+\sum\limits_{m\in\mathcal{M}}{\sum\limits_{v\in\mathcal{U}/\left\{u\right\}}{\sqrt{p_{v}^{ul}}b_{m,u}\mathbf{h}_{m,v}{\mathbf{u}_{m,u}^{\mathsf{H}}}s_{v,t}^{ul}}}\\
		&\quad+\sum\limits_{m\in\mathcal{M}}{b_{m,u}\mathbf{z}_{m}^{ul}{\mathbf{u}_{m,u}^{\mathsf{H}}}},
		\end{aligned}
		\label{3}
	\end{equation}
	where  $\widetilde{s}_{u,t}^{ul}\in\mathbb{R}^{1\times{}1}$ represents the estimated transmitted symbol. 
	Accordingly, the uplink SINR of UE $u$ is expressed as
	\begin{equation}
		\xi_{u}^{ul}=\frac{p_{u}^{ul}\left|\sum\limits_{m\in\mathcal{M}}b_{m,u}\mathbf{h}_{m,u}\mathbf{u}_{m,u}^{\mathsf{H}}\right|^2}{\sum\limits_{v\in\mathcal{U}/\left\{u\right\}}{p_{v}^{ul}\left|\sum\limits_{m\in\mathcal{M}}{b_{m,u}\mathbf{h}_{m,v}}\mathbf{u}_{m,u}^{\mathsf{H}}\right|^2}+\sum\limits_{m\in\mathcal{M}}{b_{m,u}^2}\sigma^2}.
		\label{4}
	\end{equation}
	Each UE is subject to the uplink power constraint $(C1)$ $0\leq{}p_u^{ul}\leq{}P^{ul},~\forall{}u\in\mathcal{U}$.
	The corresponding uplink rate is
	\begin{equation}
		\gamma_{u}^{ul}=w\log_2\left(1+\xi_{u}^{ul}\right),
		\label{5}
	\end{equation}
	where $w$ denotes the system bandwidth.
	
	After the CPU completes the model aggregation process, each AP $m$  transmits the updated sub-model to its associated UE $u$. The received signal of UE $u$ in the downlink time slot $t$ is expressed as
	\begin{equation}
		\begin{aligned}
		y_{u,t}^{dl}&=\sum\limits_{m\in\mathcal{M}}{\mathbf{h}_{m,u}\mathbf{v}_{m,u}^{\mathsf{H}}b_{m,u}s_{u,t}^{dl}}\\
		&\quad+\sum\limits_{m\in\mathcal{M}}{\sum\limits_{v\in\mathcal{U}/\left\{u\right\}}{\mathbf{h}_{m,u}\mathbf{v}_{m,v}^{\mathsf{H}}b_{m,v}s_{v,t}^{dl}}}+z_u^{dl}\\
		&=\mathbf{h}_{u}\mathbf{B}_{u}\mathbf{v}_{u}s_{u,t}^{dl}+\sum\limits_{v\in\mathcal{U}/\left\{u\right\}}{\mathbf{h}_{u}\mathbf{B}_{v}\mathbf{v}_{v}s_{v,t}^{dl}}+z_{u}^{dl},
		\end{aligned}
		\label{6}
	\end{equation}
	where $\mathbf{h}_{m,u}\in\mathbb{C}^{1\times{L}}$ represents the channel from AP $m$ to UE $u$, 
	$\mathbf{v}_{m,u}\in\mathbb{C}^{1\times{L}}$ represents the beamforming vector of AP $m$,
	$s_{u,t}^{dl}\in\mathbb{R}^{1\times{}1}$ represents the symbol needed by UE $u$ in the downlink time slot $t$, $z_{u}^{dl}\sim\mathcal{CN}\left(0,\sigma^2\right)$ represents the channel Gaussian white noise with power $\sigma^2$,
	$\mathbf{h}_{u}=\left[\mathbf{h}_{1,u},...,\mathbf{h}_{M,u}\right]\in\mathbb{C}^{1\times{ML}}$ is the concatenated vector of channels from APs to UE $u$,
	$\mathbf{B}_{u}=\left(\mathrm{diag}\left(b_{1,u},...,b_{M,u}\right)\otimes{}\mathbf{I}_{\mathrm{L}}\right)\in\mathbb{B}^{ML\times{}ML}$ is user association block diagonal matrix,
	and $\otimes{}$ is the Kronecker product, and
	$\mathbf{v}_{u}=\left[\mathbf{v}_{1,u},...,\mathbf{v}_{M,u}\right]^{\mathsf{H}}\in\mathbb{C}^{ML\times{}1}$ is the concatenated vector of APs' beamforming vector.
	
	Thus, the downlink SINR of UE $u$ is
	\begin{equation}
		\xi_{u}^{dl}=\frac{\mathbf{v}_{u}^{\mathsf{H}}\mathbf{B}_{u}\mathbf{h}_{u}^{\mathsf{H}}\mathbf{h}_{u}\mathbf{B}_{u}\mathbf{v}_{u}}{\sum\limits_{v\in\mathcal{U}/\left\{u\right\}}{\mathbf{v}_{v}^{\mathsf{H}}\mathbf{B}_{v}\mathbf{h}_{u}^{\mathsf{H}}\mathbf{h}_{u}\mathbf{B}_{v}\mathbf{v}_{v}}+\sigma^2}.
		\label{7}
	\end{equation}
	Here, the downlink power coefficient has been merged into the beamforming variable, which leads to a new constraint $(C2)$ that $\sum\limits_{u\in\mathcal{U}}{\|\mathbf{v}_{m,u}\|_2^2}\leq{}P^{dl},\forall{}m\in\mathcal{M}$.
	The downlink rate of UE $u$ is expressed as
	\begin{equation}
		\gamma_{u}^{dl}=w\log_2{\left(1+\xi_{u}^{dl}\right)}.
		\label{8}
	\end{equation}

\subsection{User-Centric Split Federated Learning Model in Cell-Free MIMO System}
	
	\begin{figure*}[htbp]
		\centering
		\subfloat[Interaction between UEs and APs. ]{\includegraphics[width=0.3\linewidth]{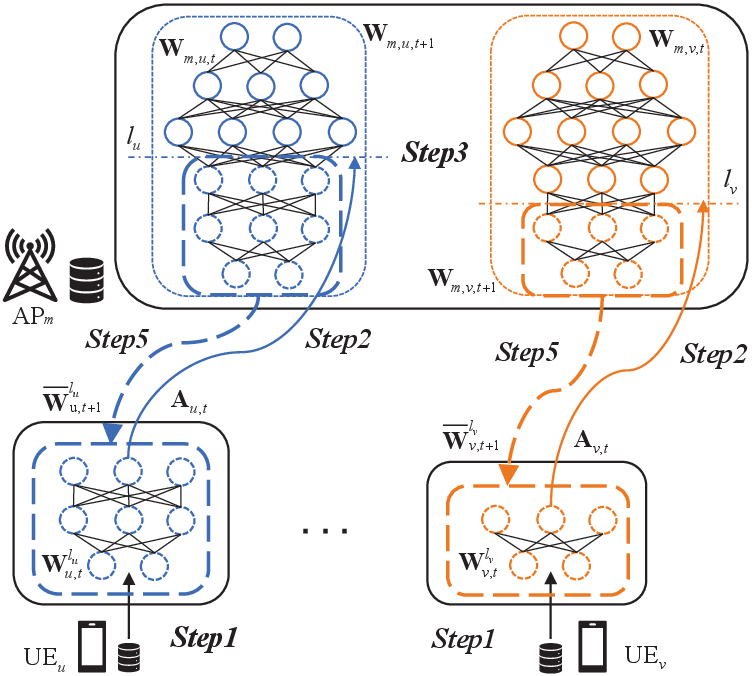}\label{fig11-1}}\hspace{6pt}
		\subfloat[Interaction between APs and CPU in $\boldsymbol{Step4}$.]{\includegraphics[width=0.65\linewidth]{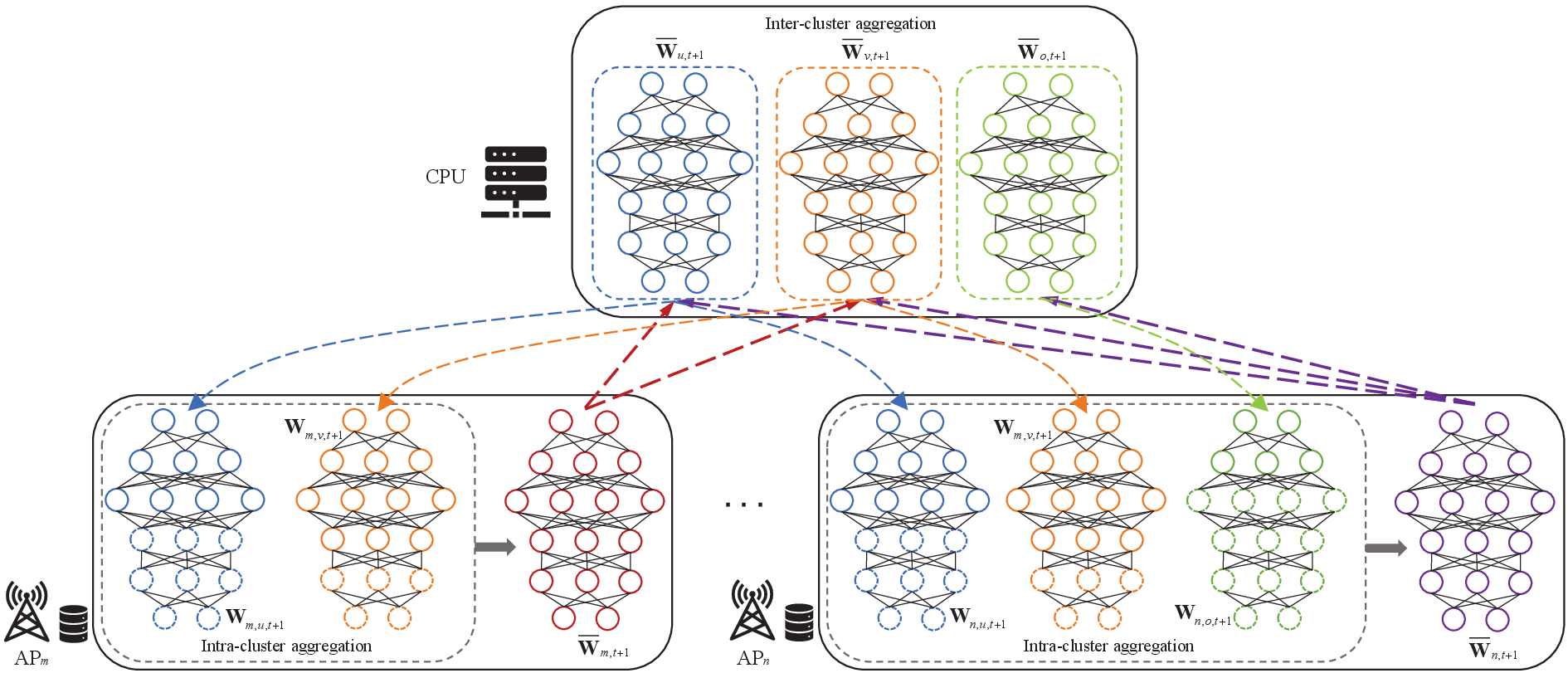}\label{fig11-2}}\hspace{6pt}
		
		\caption{A user-centric split federated learning (UCSFL) framework.}
		\label{2_2}
	\end{figure*}
	
	Given that the computing frequency of the AP-side DPUs is much higher than that of UE, model splitting presents an effective approach to reduce training latency. 
	Furthermore, the inherently user-centric nature of the considered CF-MIMO system calls for a more flexible FL aggregation method to accommodate the diverse requirements of individual user models.
	Motivated by these observations, we develop a user-centirc split federated learning (UCSFL) framework tailored for user-centirc cell-free MIMO system.
	
	Specifically, as shown in Fig. \ref{2_2}, an iteration round of UCSFL consists of five steps as follows:
	
	\begin{itemize}
		\item $\boldsymbol{Step1}$ : Each UE $u$ uses its local dataset to perform a forward propagation through its sub-model, thereby generating the activation data.
		\item $\boldsymbol{Step2}$ : Each UE $u$ transmits the activation data to its associated APs for further computation at DPUs.
		\item $\boldsymbol{Step3}$ : Each DPU $m$ performs forward propagations through the remaining sub-models of its associated UEs and subsequently updates their individual complete models through backward propagations.
		\item $\boldsymbol{Step4}$ : A two-level model aggregation is performed: intra-cluster aggregation of UEs' complete models within each AP, followed by inter-cluster aggregation of APs' aggregated models for each UE.
		\item $\boldsymbol{Step5}$ : Each AP $m$ splits the final aggregated models according to the corresponding UEs’ splitting strategies and distributes the resulting sub-models to the associated UEs.
	\end{itemize}
	
	In $\boldsymbol{Step1}$, the activation data of UE $u$ in episode $t$ is expressed as 
	\begin{equation}
		\mathbf{A}_{u,t}=p_u(\mathbf{{W}}_{u,t}^{l_u}~|~\zeta_{u,t}),
		\label{9}
	\end{equation}
	where $p_u(\cdot)$ is the forward propagation function, $\mathbf{{W}}_{u,t}^{l_u}$ denotes the network parameters of the front $l_u$ layers, and $\zeta_{u,t}$ is the input data.
	The activation data $\mathbf{A}_{u,t}$ is then transmitted to the associated APs in $\boldsymbol{Step2}$ for subsequent processing at the DPUs.
	
	In $\boldsymbol{Step3}$, the further forward propagation through the remaining sub-model of UE $u$ in episode $t$ is expressed as
	\begin{equation}
		F_{m,u,t}=P_u(\mathbf{W}_{m,u,t}^{L-l_u}~|~\mathbf{A}_{u,t}),
		\label{10}
	\end{equation}
	where $F_{m,u,t}$ is the value of loss function, $P_u(\cdot)$ is the further forward propagation function, and $\mathbf{W}_{m,u,t}^{L-l_u}$ denotes the network parameters of the remaining $L-l_u$ layers.
	Additionally, the complete model update of UE $u$ in episode $t$ is expressed as
	\begin{equation}
		\mathbf{W}_{m,u,t+1}=\mathbf{W}_{m,u,t}-\eta_{t}\nabla{}F_{m,u,t}\left(\mathbf{W}_{m,u,t}~|~\zeta_{u,t}\right),
		\label{11}
	\end{equation}
	where $\eta_{t}$ is the learning rate, and $\nabla{}F_{m,u,t}(\cdot)$ is the gradient of complete model.
	It is worth noting that the model storage schemes differ between APs and UEs.
	The APs maintain and update the complete model for further aggregation in $\boldsymbol{Step4}$, whereas UEs only store and update their local sub-models.
	
	In $\boldsymbol{Step4}$, the intra-cluster aggregation within AP $m$ is expressed as
	\begin{equation}
		\overline{\mathbf{W}}_{m,t+1}=\frac{1}{C_m}\sum\limits_{v\in\bm{\psi}_m}{b_{m,v}\mathbf{W}_{m,v,t+1}},
		\label{12}
	\end{equation}
	where $C_m=\sum\limits_{v\in\mathcal{U}}{b_{m,v}}$ denotes the size of 
	the UE cluster $\bm{\psi}_m$ associated with AP $m$.
	Meanwhile, the inter-cluster aggregation for UE $u$ is expressed as
	\begin{equation}
		\overline{\mathbf{W}}_{u,t+1}=\frac{1}{C_{u}}\sum\limits_{m\in\bm{\psi}_u}{b_{m,u}\overline{\mathbf{W}}_{m,t+1}},
		\label{}
	\end{equation}
	where $C_u=\sum\limits_{m\in\mathcal{M}}{b_{m,u}}$ denotes the size of 
	the AP cluster $\bm{\psi}_u$ associated with UE $u$.
	
	In $\boldsymbol{Step5}$, each AP splits the user-specific aggregated models according to its associated UEs' splitting strategies, thus generating their corresponding sub-models, namely
	\begin{equation}
		\mathbf{W}_{u,t+1}^{l_u}=\left[\overline{\mathbf{W}}_{u,t+1}\right]_{l_u},
		\label{}
	\end{equation}
	where $[\cdot]_{l_u}$ is the model splitting function with the parameter $l_u$.	
	The resulting sub-models are then transmitted to the corresponding UEs for the next iteration.

\subsection{Convergence Analysis}
	\label{Convergence Analysis}
	To characterize the learning behavior of the proposed UCSFL framework, we conduct a fundamental convergence analysis. For analytical tractability, the following standard assumptions are adopted.
	\begin{assumption}
		The loss function $F$ is Lipschitz smooth, i.e.,  $F_{m,u}\left(\mathbf{v}\right)-F_{m,u}\left(\mathbf{w}\right)\leq\left(\mathbf{v}-\mathbf{w}\right)^{\mathsf{T}}\nabla{}F_{m,u}\left(\mathbf{w}\right)+\frac{\beta}{2}\|\mathbf{v}-\mathbf{w}\|^2$.
	\end{assumption}
	
	\begin{assumption}
		The loss function $F$ is $\mu$-strongly convex, i.e.,
		$F_{m,u}\left(\mathbf{v}\right)-F_{m,u}\left(\mathbf{w}\right)\geq\left(\mathbf{v}-\mathbf{w}\right)^{\mathsf{T}}\nabla{}F_{m,u}\left(\mathbf{w}\right)+\frac{\mu}{2}\|\mathbf{v}-\mathbf{w}\|^2$.
	\end{assumption}
	
	\begin{assumption}
		The variance of the gradients for each layer in device has upper bound:
		
		\qquad\quad$\mathbb{E}\|\nabla{}F_{m,u}\left(\mathbf{w},\zeta_{u,t}\right)-\nabla{}F_{m,u}\left(\mathbf{w}\right)\|^2\leq{}L\epsilon^2$.
	\end{assumption}
	
	\begin{assumption}
		The expected squared norm of each layer’s gradients has upper bound:
		
		\qquad\qquad\qquad$\mathbb{E}\|\nabla{}F_{m,u}\left(\mathbf{w},\zeta_{u,t}\right)\|^2\leq{}LZ^2$.
	\end{assumption}
	
	\begin{assumption}
		For UE $u$'s AP-cluster $\bm{\psi}_u$, the probability of each AP being selected to associated with UE $u$ is $P\left(m\in\bm{\psi}_u\right)=\frac{C_u}{M}$.
	\end{assumption}
	
	\begin{assumption}
		For AP $m$'s UE-cluster $\bm{\psi}_m$, the probability of each UE being selected to associated with AP $m$ is $P\left(u\in\bm{\psi}_m\right)=\frac{C_m}{U}$.
	\end{assumption}
	
	Under these assumptions, we derive an upper bound on the expected optimality gap of the loss function after $T$ iterations, as stated in the following theorem.
	\begin{theorem}
		Let $\alpha=\frac{\beta}{\mu}$, $\iota=\max\left\{8\alpha, T\right\}-1$ and choose the learning rate $\eta_t=\frac{2}{\mu\left(t+\iota\right)}$
		. After $T$ iterations, there is
		an upper bound between the mean value and the optimal
		value of any device's loss function, which satisfies the following
		relationship:
		\begin{equation}
		\begin{aligned}
			&\mathbb{E}\left[F\left(\overline{\mathbf{W}}_{u,T}\right)\right]-F^{\ast}\\
			&\leq\frac{\alpha}{t+\iota}\left[\frac{2R}{\mu}+\frac{\mu\left(\iota+1\right)}{2}\mathbb{E}\vert\vert {\overline{\mathbf{W}}_{u,1}-\mathbf{W}^{\ast}} \vert\vert^2\right],
		\end{aligned}
		\label{eq15}
		\end{equation}
		where $R=\frac{3C_u+1}{C_u}\left(T-1\right)^2LZ^2+\frac{C_u+1}{2C_u}LZ^2+\frac{C_u+1}{2C_u}L\epsilon^2+2\beta\Gamma$, 
		and  $\Gamma=\mathbb{E}\left(F^{\ast}-F_{u,t}^{\ast}\right)$ denotes the	impact of data heterogeneity.
		\label{th1}
	\end{theorem}

	\begin{proof}

		Seen Appendix A.
		
	\end{proof}

	From Eq. \eqref{eq15}, we take the partial derivative of the upper bound with respect to $C_u$ and have
	\begin{equation}
		\frac{\partial{R}}{\partial{C_u}}=-\frac{1}{C_u^2}\left[\left(T-1\right)^2LZ^2+\frac{L}{2}\left(Z^2+\epsilon^2\right)\right]<0,
		\label{eq16}
	\end{equation} 
	where $C_u\geq1$ ensures that each UE $u$ is associated to at least one AP, otherwise it does not participate in the UCSFL process.
	Therefore, a larger AP-cluster size leads to a lower upper bound, which implies higher model accuracy.
	Intuitively, this improvement arises because a larger AP cluster enables aggregation over a more diverse set of intermediate models. Consequently, AP clustering plays a critical role in determining the training accuracy of UCSFL in user-centric CF-MIMO systems.

\section{Joint Latency and Accuracy for UCSFL: Problem Formulation}
\label{Problem Formulation}

For the UCSFL training problem, our objective is to minimize the training latency while maximizing the training accuracy. 
These two objectives, are inherently conflicting in practical systems, thereby necessitating a joint optimization framework that explicitly captures their trade-off.
According to Eq. \eqref{eq16}, the optimality gap of loss value is closely related to the AP-cluster size $C_u$ of each UE, which serves as a key parameter reflecting the achievable training accuracy.
Motivated by this observation, we characterize the latency–accuracy trade-off by defining an objective that relates training latency to the AP-cluster size, and formulate a joint latency and accuracy optimization problem for the UCSFL framework.
	 
\subsection{Training Latency}

	Due to potentially different model splitting points across UEs, the forward propagation time of local sub-models may vary, which results in asynchronous uplink transmission start times. Consequently, the inter-user interference experienced by a UE during uplink transmission is time-varying and depends on the set of concurrently transmitting UEs. Owing to this complexity, the exact uplink transmission latency is difficult to characterize. Instead, we consider an upper bound on the uplink transmission latency of UE $u$, expressed as
	\begin{equation}
		t_{u}^{ul}=\frac{D_u}{\gamma_{u}^{ul}},
		\label{}
	\end{equation}
	where $D_u$ is the data size of activation data $\mathbf{A}_{u}$ and $\gamma_u^{ul}$ is the lower bound of UE $u$'s real uplink transmission rate that suffers the maximum inter-user interference.
	This method mitigates the impact of varying local computing latency among UEs on the uplink transmission by relaxing the transmission latency to its upper bound.

	In this case, the relaxed uplink transmission latency can be interpreted as the latency incurred when all UEs initiate uplink transmissions simultaneously.
	Accordingly, the UE-side local computing latency is determined by the slowest UE to complete forward propagation, which is given by
	\begin{equation}
		t^{ue}=\max\limits_{l_u}\left\{\frac{x_{l_u}c}{f^{ue}}\right\},
		\label{}
	\end{equation}
	where $c$ is the number of computing cycles for each computing load, and $x_{l_u}$ is the amount of computing load of the front $l_u$ layers of the model on UE $u$. This load can be quantified by the number of multiply–accumulate operations \cite{liu2018demand, xu2019reform}, expressed as
	\begin{equation}
		x_{l_u}=\sum\limits_{i=1}^{l_u}{\sum\limits_{j=1}^{n_i}{r_i^jq_i^jn_{i-1}h_i^jw_i^jn_ik}},
		\label{}
	\end{equation}
	where $r_i^j$ and $q_i^j$ are the $j$-th filter’s kernel sizes of layer $i$, $n_i$ is the number of kernels in layer $i$, $h_i^j$ and $w_i^j$ are the corresponding height and width of output feature map
	respectively, and $k$ is the batch size.
	Therefore, the size of the activation data generated by UE $u$ is expressed as
	\begin{equation}
		D_u=h_{l_u}w_{l_u}n_{l_u}k.
		\label{}
	\end{equation}
	
	After all APs finishing receiving the activation data, the corresponding DPUs perform forward propagation through the remaining sub-models, followed by backward propagation of the complete models prior to intra-cluster aggregation.
	The forward propagation latency at DPU $m$ can be expressed as
	\begin{equation}
		t_{m}^{dpu}=\max_u\left\{\frac{b_{m,u}\left(x_L-x_{l_u}\right)c}{f^{dpu}}\right\},
		\label{}
	\end{equation}
	where $x_L$ represents the amount of computing load of all layers.
	Since backward propagation is executed sequentially without interruption at each DPU, its latency is identical across DPUs. We therefore model the DPU-side backward propagation latency as a constant $t^{back}$.
	
	Given the ideal wired backhaul assumption, the latency of wired transmission is neglected. In addition, the latencies associated with intra-cluster aggregation, inter-cluster aggregation, and model splitting are assumed to be negligible.
	Hence, all downlink transmissions start simultaneously, and the downlink transmission latency of UE $u$ is expressed as
	\begin{equation}
		t_{u}^{dl}=\frac{D_{l_u}}{\gamma_{u}^{dl}},
		\label{}
	\end{equation}
	where $D_{l_u}$ is the data size of the sub-model $\mathbf{W}_{m,u}^{l_u}$ transmitted to UE $u$.
	
	Combining the above components, an upper bound on the total training latency of UE $u$ per iteration is given by 
	\begin{equation}
		t_u=t^{ue}+\max\limits_u\left\{t_u^{ul}\right\}+\max_m\left\{t_m^{dpu}\right\}+t^{back}+t_u^{dl},
		\label{}
	\end{equation}
	where the maximization of $t_u^{l_u}$ and $t_m^{L-l_u}$ ensures synchronization among all UEs before the inter-cluster aggregation phase.
	Since all DPUs operate at the same computing frequency, the DPU-side latency simplifies to
	\begin{equation}
		\max_m\left\{t_m^{dpu}\right\}=\max_{l_u}\left\{\frac{\left(x_L-x_{l_u}\right)c}{f^{dpu}}\right\}.
		\label{}
	\end{equation}
	This indicates that the maximum DPU-side forward propagation latency depends solely on the model splitting strategies adopted by UEs.

\subsection{Problem Formulation}
	Motivated by the trade-off between model training latency and accuracy, we define a new performance metric for UE $u$, termed the latency-to-accuracy ratio, as
	\begin{equation}
		\Upsilon_u=\frac{\mathbb{E}\{t_u\}}{{C_u}^{\ell}},
		\label{}
	\end{equation}
	where $\ell{}>0$ is a weighting coefficient that modulates the relative importance of training accuracy.
	Our overall goal is to minimize the latency-accuracy ratio. This is achieved by jointly reducing the expected UE's training latency $\mathbb{E}\{t_u\}$ and increasing the size of AP-cluster $C_u$, which is positively correlated with training accuracy as established in Section \ref{Convergence Analysis}.
	
	Rather than directly optimizing the performance of individual UEs, we aim to minimize the upper bound of latency across all UEs.
	This ensures fairness among UEs and prevents any single UE from experiencing extremely poor performance.
	Accordingly, the joint optimization problem of model splitting, user association, uplink power control, and downlink beamforming is formulated as
	\begin{equation} \label{26}
		\begin{aligned}
			&\min\limits_{\substack{\left\{l_u\right\},\left\{b_{m,u}\right\},\\ \left\{p_u^{ul}\right\},\left\{\mathbf{v}_{m,u}\right\}}}{\max\limits_{u}{~\Upsilon_u}} \\
			&\qquad\mathrm{s.t.}~(C1),(C2), \\
			&\qquad~~~~~l_u\in\left\{1,2,...,L\right\},~\forall{}u,\\
			&\qquad~~~~~b_{m,u}\in\left\{0,1\right\},~\forall{}u,m.
		\end{aligned}
	\end{equation}
	The resulting problem is a non-convex mixed-integer optimization problem that explicitly captures the trade-off between training latency and accuracy in UCSFL-enabled user-centric CF-MIMO networks.
	Given the strong variable coupling and complex objective, the problem is highly nontrivial. 
	We tackle it by first temporally decoupling the original problem, then designing tailored algorithms for each decoupled subproblem.

\section{Joint Communication and Computation Resource Allocation: Proposed Algorithm}
\label{Proposed Algorithm}

In this section, we decompose the joint optimization problem into two sub-problems that are tackled at different time scales.
Specifically, the short-term sub-problem focuses on the joint optimization of uplink transmission power and downlink beamforming under fast-varying wireless channels, which is addressed using convex optimization techniques. On this basis, the long-term sub-problem optimizes user association and model splitting strategies under relatively slow-varying system dynamics, which is solved using a multi-agent deep reinforcement learning approach.

\subsection{Problem Splitting}

To guarantee the correct execution of both the split forward propagation and the two-level aggregation, the UE-clusters, AP-clusters, and model splitting strategies are fixed within each UCSFL training iteration.
In contrast, transmission parameters, i.e., transmission power and beamforming, must be dynamically adjusted to adapt to instantaneous wireless channel variations during each iteration.
This separation of time scales enables the original joint problem to be decomposed into a series of short-term sub-problems and a long-term sub-problem as follows.

For given long-term model splitting and user association strategies, one short-term sub-problem is simplified as
\begin{equation}
	\begin{aligned}
		&\min\limits_{\substack{\left\{p_u^{ul}\right\},
				\left\{\mathbf{v}_{m,u}\right\}}}{\max\limits_u\left\{t_u^{ul}\right\}+\max\limits_u\left\{t_u^{dl}\right\}} \\
		&\qquad\mathrm{s.t.}~(C1),(C2).
	\end{aligned}
	\label{27}
\end{equation} 
Given the sub-optimal solutions $\left\{p_u^{ul}\right\}$ and $\left\{\mathbf{v}_{m,u}\right\}$ obtained from a series of short-term sub-problems, the long-term sub-problem is expressed as
\begin{equation}
	\begin{aligned}
		&\min\limits_{\substack{\left\{l_u\right\},\left\{b_{m,u}\right\}}}{\quad\max\limits_{u}{~\Upsilon_u}} \\
		&~~~~\mathrm{s.t.}~l_u\in\left\{1,2,...,L\right\},~\forall{}u,\\
		&\qquad~~~b_{m,u}\in\left\{0,1\right\},~\forall{}u,m,
	\end{aligned}
	\label{28}
\end{equation}
where the objective function remains unchanged.

\subsection{Solving the Short-Term Sub-problem}

To handle the max operators in the short-term sub-problem, we introduce auxiliary variables $t_{\rm{max}}^{\rm{ul}}$ and $t_{\rm{max}}^{\rm{dl}}$, yielding the equivalent formulation
\begin{equation}
	\begin{aligned}
		&\max\limits_{\substack{\left\{p_u^{ul}\right\},\left\{\mathbf{v}_{m,u}\right\}}}
			{-t_{\rm{max}}^{\rm{ul}}-t_{\rm{max}}^{\rm{dl}}} \\
		&\qquad\mathrm{s.t.}~(C1),(C2),\\
		&\qquad~~~~~\frac{D_u}{t_{\rm{max}}^{\rm{ul}}}\leq{}\gamma_{u}^{\rm{ul}},~\forall{}u,\\
		&\qquad~~~~~\frac{D_{l_u}}{t_{\rm{max}}^{\rm{dl}}}\leq{}\gamma_{u}^{\rm{dl}},~\forall{}u,
	\end{aligned}
	\label{29}
\end{equation} 
where $t_{\rm{max}}^{\rm{ul}}$ and $t_{\rm{max}}^{\rm{dl}}$ represent the maximum latencies for uplink and downlink transmissions among all UEs.

On this basis, constraints $(C3)$ and $(C4)$ are incorporated into the objective function by introducing penalty variables $\bm{\theta}=\left\{\theta_u\right\}_{u\in\mathcal{U}}$ and $\bm{\pi}=\left\{\pi_u\right\}_{u\in\mathcal{U}}$. 
Meanwhile, auxiliary variables $\bm{\omega}=\left\{\omega_u\right\}_{u\in\mathcal{U}}$ and $\bm{\Omega}=\left\{\Omega_u\right\}_{u\in\mathcal{U}}$ are introduced to leverage fractional programming (FP) in \cite{shen2018fractional1} and \cite{shen2018fractional2}, thereby constructing the Lagrange function $L_1$ as given in Eq. \eqref{30} at the bottom of the next page,
\begin{figure*}[hbp] % hb底部，ht为头部
	\centering % 公式居中
	\hrulefill % 添加一条水平线
	\vspace*{1pt} % 调整线与公式之间的距离
	\begin{equation}
		\begin{aligned}
		&L_1\left(\mathbf{p}^{\rm{ul}},\mathbf{v},\bm{\theta},\bm{\pi},\bm{\omega},\bm{\Omega}\right)\\
		&=-t_{\rm{max}}^{\rm{ul}}-t_{\rm{max}}^{\rm{dl}}-\sum\limits_{u\in\mathcal{U}}{\left(\theta_u\frac{D_u}{t_{\rm{max}}^{\rm{ul}}}+\pi_u\frac{D_{l_u}}{t_{\rm{max}}^{\rm{dl}}}\right)}+\sum\limits_{u\in\mathcal{U}}{\left[\theta_uw\log_2{\left(1+\xi_u^{ul}\right)}-\frac{\theta_uw}{\ln2}\xi_{u}^{ul}+\pi_uw\log_2\left(1+\xi_u^{dl}\right)-\frac{\pi_uw}{\ln2}\xi_{u}^{dl}\right]}\\
		&\quad+\frac{\mathrm{w}}{\ln2}\sum\limits_{u\in\mathcal{U}}{\left[2\omega_u\sqrt{\theta_u\left(1+\xi_u^{ul}\right)P_u^{\rm{US}}}-\omega_u^2\left(P_u^{\rm{US}}+P_u^{IUI}+C_u\sigma^2\right)\right]}\\
		&\quad+\frac{\mathrm{w}}{\ln2}\sum\limits_{u\in\mathcal{U}}{\left[2\sqrt{\pi_u\left(1+\xi_u^{dl}\right)}\mathrm{Re}\left\{\Omega_u^{\ast}\mathbf{h}_{u}\mathbf{B}_u\mathbf{v}_u\right\}-\left|\Omega_u\right|^2\left(\sum\limits_{v\in\mathcal{U}}{\mathbf{v}_{v}^{\mathsf{H}}\mathbf{B}_{v}\mathbf{h}_{u}^{\mathsf{H}}\mathbf{h}_{u}\mathbf{B}_{v}\mathbf{v}_{v}}+\sigma^2\right)\right]},
		\end{aligned}
		\label{30}
	\end{equation}
\end{figure*}
where $P_u^{\rm{US}}=p_{u}^{ul}\left|\sum\limits_{m\in\mathcal{M}}{b_{m,u}}\mathbf{h}_{m,u}\mathbf{u}_{m,u}^{\mathsf{H}}\right|^2$ represents the useful signal power,
$P_u^{IUI}=\sum\limits_{v\in\mathcal{U}/\left\{u\right\}}{p_{v}^{ul}\left|\sum\limits_{m\in\mathcal{M}}{b_{m,u}\mathbf{h}_{m,v}\mathbf{u}_{m,u}^{\mathsf{H}}}\right|^2}$ represents the inter-user inference signal power,
$\mathrm{Re}\left\{\cdot\right\}$ denotes the real part of a complex number , and the operator
$\left(\cdot\right)^{\ast}$ denotes conjugate.

Then, we take the partial derivative of Eq. \eqref{30} with respect to $\omega_u$, and obtain 
\begin{equation}
	\omega_u=\frac{\sqrt{\theta_u\left(1+\xi_u^{ul}\right)P_u^{\rm{US}}}}{P_u^{\rm{US}}+P_u^{\rm{IUI}}+C_u\sigma^2}.
	\label{31}
\end{equation}
Meanwhile, $\Omega_u$ can be obtained from its corresponding first optimality condition as
\begin{equation}
	\Omega_u=\frac{\sqrt{\pi_u\left(1+\xi_u^{dl}\right)}\mathbf{h}_u\mathbf{B}_u\mathbf{v}_{u}}{\sum\limits_{u\in\mathcal{U}}{\mathbf{v}_{v}^{\mathsf{H}}\mathbf{B}_{v}\mathbf{h}_{u}^{\mathsf{H}}\mathbf{h}_{u}\mathbf{B}_{v}\mathbf{v}_{v}}+\sigma^2}.
	\label{32}
\end{equation}

By incorporating the power constraints $(C1)$ and $(C2)$, the Lagrange function $L_2$ is constructed as
\begin{equation}
	\begin{aligned}
	&L_2\left(\mathbf{p}^{\rm{ul}},\mathbf{v},\bm{\theta},\bm{\pi},\bm{\upsilon},\bm{\chi}\right)\\
	&=L_1\left(\mathbf{p}^{\rm{ul}},\mathbf{v},\bm{\theta},\bm{\pi}\right)-
	\sum\limits_{u\in\mathcal{U}}{\upsilon_u\left(p_{u}^{\mathrm{ul}}-P^{\mathrm{ul}}\right)}\\
	&\quad-\sum\limits_{u\in\mathcal{U}}{\mathbf{v}^{\mathsf{H}}_{u}\bm{X}\mathbf{v}_u}+\sum\limits_{m\in\mathcal{M}}{\chi_mP^{\mathrm{dl}}},
	\end{aligned}
	\label{33}
\end{equation}
where $\bm{\upsilon}=\left\{\upsilon_u\right\}_{u\in\mathcal{U}}$ and $\bm{\chi}=\left\{\chi_m\right\}_{m\in\mathcal{M}}$ are the Lagrange multipliers corresponding to the constraints in $(C1)$ and $(C2)$ respectively, $\bm{X}=\left(\mathrm{diag}\left(\left\{\chi_m\right\}_{m\in\mathcal{M}}\right)\otimes{}\mathbf{I}_{\mathrm{L}}\right)\in\mathbb{R}^{ML\times{}ML}$ represents the Lagrange multiplier block diagonal matrix.

Then, we take the partial derivative of Eq. \eqref{33} with respect to $p_u^{\mathrm{ul}}$, and obtain
\begin{equation}
	p_u^{\mathrm{ul}}=\frac{\theta_u\left(1+\xi_u^{\mathrm{ul}}\right)\omega_u^2\left|\sum\limits_{m\in\mathcal{M}}{b_{m,u}\mathbf{h}_{m,u}\mathbf{u}_{m,u}^{\mathsf{H}}}\right|^2}{\left[\sum\limits_{v\in{U}}{\omega_v^2\left|\sum\limits_{m\in\mathcal{M}}{b_{m,v}\mathbf{h}_{m,u}\mathbf{u}_{m,v}^{\mathsf{H}}}\right|^2}+\frac{\ln2}{w}\upsilon_u\right]^2},
	\label{34}
\end{equation}
where Lagrange multipliers $\bm{\upsilon}$ are obtained by the binary search method in order to make $\mathbf{p}^{\mathrm{ul}}$ satisfy the power constraint $(C1)$ before updating the power expression.
Meanwhile, the optimal beamforming expression for $\mathbf{v}_u$ can be constructed as
\begin{equation}
	\begin{aligned}
	\mathbf{v}_u^{}=&\left(\sum\limits_{v\in\mathcal{U}}{\left|\Omega_{v}\right|^2\mathbf{B}_{u}\mathbf{h}_{v}^{\mathsf{H}}\mathbf{h}_{v}\mathbf{B}_{u}}+\frac{\ln2}{w}\bm{X}\right)^{-1}\\
	&\times\Omega_u\sqrt{\pi_u\left(1+\xi_u^{\mathrm{dl}}\right)}\mathbf{B}_u\mathbf{h}_u^{\mathsf{H}},
	\end{aligned}
	\label{35}
\end{equation}
where $\left(\cdot\right)^{-1}$ denotes the inverse matrix and any Lagrange multiplier $\chi_m$ is obtained by the binary search method in order to satisfy the constraint $(C2)$ before updating all beamforming expressions $\mathbf{v}$. 

In addition, for penalty factors $\bm{\theta}$ and $\bm{\pi}$, we can update them by sub-gradients considering that there exists a coupling relationship between variables with different subscripts. 
The sub-gradient of $\theta_u$ can be calculated as
\begin{equation}
	g_u^{\theta}=\gamma_u^{\mathrm{ul}}-\frac{D_u}{t_{\mathrm{max}}^{\mathrm{ul}}},
	\label{36}
\end{equation}
where $t_{\mathrm{max}}^{\mathrm{ul}}$ is determined by the result of the last update, i.e.,
\begin{equation}
	\left(t_{\mathrm{max}}^{\mathrm{ul}}\right)^{n}=
	\frac{D_u}{\min\limits_{u\in\mathcal{U}}{\left(\left(\gamma_u^{\mathrm{ul}}\right)^{n-1}\right)}}.
	\label{37}
\end{equation}
Here, $n$ is the number of update times.
Thus, the update of $\theta_u$ can be expressed as
\begin{equation}
	\theta_u=
	\begin{cases}
		\theta_u,~~~~~~~~~~~\theta_u-\rho^{\theta}g_u^{\theta}<0\\
		\theta_u-\rho^{\theta}g_u^{\theta},~~\theta_u-\rho^{\theta}g_u^{\theta}\geq0
	\end{cases},
	\label{38}
\end{equation}
where $\rho^{\theta}$ is a constant step size for updating penalty factor $\bm{\theta}$.

Similarly, the sub-gradient of $\pi_u$ can be calculated as
\begin{equation}
	g_u^{\pi}=\gamma_u^{\mathrm{dl}}-\frac{D_{l_{u}}}{t_{\mathrm{max}}^{\mathrm{dl}}},
	\label{39}
\end{equation}
where $t_{\mathrm{max}}^{\mathrm{dl}}$ is determined by the result of the last update, namely
\begin{equation}
	\left(t_{\mathrm{max}}^{\mathrm{dl}}\right)^{n}=
	\frac{D_{l_u}}{\min\limits_{u\in\mathcal{U}}{\left(\left(\gamma_u^{\mathrm{dl}}\right)^{n-1}\right)}}.
	\label{40}
\end{equation}
Thus, the update of $\pi_u$ can be expressed as
\begin{equation}
	\pi_u=
	\begin{cases}
		\pi_u,~~~~~~~~~~~\pi_u-\rho^{\pi}g_u^{\pi}<0\\
		\pi_u-\rho^{\pi}g_u^{\pi},~~\pi_u-\rho^{\pi}g_u^{\pi}\geq0
	\end{cases},
	\label{41}
\end{equation}
where $\rho^{\pi}$ is a constant step size for updating penalty factor $\bm{\pi}$.

Based on the above updates, we adopt a nested block coordinate descent (NBCD) framework to iteratively optimize uplink power control and downlink beamforming, as summarized in Algorithm 1. Upon convergence, the optimized $\left\{p^{\rm{ul}}_u, \mathbf{v}_u|u\in\mathcal{U}\right\}$ can be used to compute the corresponding uplink and downlink transmission rates.

\begin{algorithm}[hbtp]
	\caption{~NBCD for Power Control and Beamforming}
	\begin{algorithmic}[1]
		\State \textbf{Init:}  Uplik power control $\mathbf{p}^{\rm{ul}}$, downlink beamforming $\mathbf{v}$,
		\Statex\qquad\ penalty factors $\left\{\bm{\theta},\bm{\pi}\right\}$
		\Repeat
		\Repeat
		\State Update uplink SINR $\bm{\xi}^{ul}$ by using \eqref{4}
		\State Update auxiliary variables $\bm{\omega}$ of FP by using \eqref{31}
		\State Find suitable penalty factors $\bm{\upsilon}$ by binary search
		\State Update uplink power control $\mathbf{p}^{ul}$ by using \eqref{34}
		\Until convergence
		\Repeat
		\State Update downlink SINR $\bm{\xi}^{dl}$ by using \eqref{7}
		\State Update auxiliary variables $\bm{\Omega}$ of FP by using \eqref{32}
		\State Find suitable penalty factors $\bm{\chi}$ by binary search
		\State Update downlink beamforming $\mathbf{v}$ by using \eqref{35}
		\Until convergence
		\State Calculate the sub-gradients by using \eqref{36}, \eqref{39}
		\State Update penalty factors $\left\{\bm{\theta},\bm{\pi}\right\}$ by using \eqref{38}, \eqref{41}
		\Until convergence
		\State \textbf{Output} $\left\{p^{\rm{ul}}_u, \mathbf{v}_u|u\in\mathcal{U}\right\}$
	\end{algorithmic}
	\label{algo1}
\end{algorithm}

\subsection{Solving the long-Term Sub-problem}

For the long-term sub-problem, our goal is to minimize the objective function
$\max\limits_u\left\{\Upsilon_u\left(\mathbf{l},\mathbf{b}\right)\right\}$. Due to the discrete nature of variables and the dependency of function parameters on these variables, i.e., $\left\{D_u\right\}$, $\left\{D_{l_u}\right\}$, and $\left\{x_{l_u}\right\}$, conventional convex optimization methods are not applicable. To address this challenge, we develop a multi-agent intelligent search (MAIS) algorithm.

The MAIS algorithm is inspired by the multi-agent actor-critic (MAAC) architecture in deep reinforcement learning (DRL). The original problem is first transformed into a Markov decision process (MDP), which is then solved by alternately training actor and critic networks. In this framework, each agent (i.e., UE) is equipped with an actor network, which predicts its next action based on its current state. The current states and next actions of all agents are then aggregated and fed into a shared critic network to estimate the expected return.

In order to develop the proposed MAIS algorithm, we first formalize the key components of the MDP for the DRL framework, namely, the state space, the action space, and the reward function.

\begin{itemize}
	\item \textit{State Space} $\mathcal{S}$: Based on the delay and transmission models together with the long-term sub-problem, we define the state of UE $u$ in episode $t$ as
	\begin{equation}
		\mathbf{s}_u\left(t\right)=\left\{D_u\left(t\right),D_{l_u}\left(t\right),x_{l_u}\left(t\right),\mathbb{E}\{\gamma_{u}^{ul}\left(t\right)\},\mathbb{E}\{\gamma_{u}^{dl}\left(t\right)\}\right\}.
		\label{42}
	\end{equation}
	Then, the state space in episode $t$ can be defined as
	\begin{equation}
		\mathcal{S}\left(t\right)=\left\{\mathbf{s}_u\left(t\right)|u\in\mathcal{U}\right\}.
		\label{43}
	\end{equation}
	
	\item \textit{Action Space} $\mathcal{A}$: In episode $t$, the action of UE $u$ consists of model split action $l_u\left(t\right)$ and user association action $\mathbf{b}_{u}\left(t\right)=\left\{b_{1,u}\left(t\right),...,b_{i,u}\left(t\right),...,b_{M,u}\left(t\right)\right\}$ as follows:
	\begin{equation}
		\mathbf{a}_u\left(t\right)=\left\{l_u\left(t\right),\mathbf{b}_{u}\left(t\right)\right\}.
		\label{44}
	\end{equation}
	Then, the action space in episode $t$ can be defined as
	\begin{equation}
		\mathcal{A}\left(t\right)=\left\{\mathbf{a}_u\left(t\right)|u\in\mathcal{U}\right\}.
		\label{45}
	\end{equation}
	
	\item \textit{Reward Function}: To deduce a MDP reformulation, we define the reward function in episode $t$ as
	\begin{equation}
		U_t\left(\mathcal{S}\left(t\right),\mathcal{A}\left(t\right)\right)=-\max\limits_{u}{\left\{\Upsilon_u\left(t\right)\right\}},
		\label{46}
	\end{equation}
	which means a reward $U_t\left(\mathcal{S}\left(t\right),\mathcal{A}\left(t\right)\right)$ is obtained in global state $\mathcal{S}\left(t\right)$ after executing global action $\mathcal{A}\left(t\right)$.
\end{itemize}

Following the MAAC architecture, we employ multiple actor networks to individually output its estimated action vector $\hat{\mathbf{a}}_u\left(t\right)=O_t\left(\mathbf{s}_u\left(t\right)|\bm{\psi}_u\right)$, while utilizing a single critic network to evaluate the pair of state-action, denoted as $\hat{U}_t\left(\mathcal{S}\left(t\right),\mathcal{A}\left(t\right)|\bm{\phi}\right)$.
Therefore, we can express the time division (TD) error as
\begin{equation}
	\delta_{t}^{TD}=U_t+\gamma\hat{U}_{t+1}-\hat{U}_t,
	\label{47}
\end{equation}
where $\gamma$ is the discount factor.

Meanwhile, we define the advantage function $A_t\left(\mathcal{S},\mathcal{A}\right)$ to quantify the effectiveness of executing a global action $\mathcal{A}\left(t\right)$ in a global state $\mathcal{S}\left(t\right)$.
Considering that the value of advantage function is hard to obtain, we use the generalized advantage estimation (GAE) to approximate $A_t\left(\mathcal{S},\mathcal{A}\right)$, which is expressed as
\begin{equation}
	\hat{A}_t=\sum\limits_{l=0}^{\mathbb{T}-t}{\left(\gamma\lambda\right)^{l}\delta_{t+l}^{TD}},
	\label{48}
\end{equation}
where $\mathbb{T}$ is the step size, and $\lambda$ is the GAE parameter.

For the single critic network $\bm{\phi}$, the loss function is expressed as
\begin{equation}
	L^{VF}\left(\bm{\phi}\right)=\mathbb{E}_t\left[\left(\hat{U}_t\left(\mathcal{S},\mathcal{A}|\bm{\phi}^{old}\right)+\hat{A}_t-\hat{U}_t\left(\mathcal{S},\mathcal{A}|\bm{\phi}\right)\right)^2\right],
	\label{49}
\end{equation}
where $\bm{\phi}^{old}$ is the critic network with old parameters. The update of critic network is expressed as 
\begin{equation}
	\bm{\phi}\leftarrow\bm{\phi}-\eta_{\bm{\phi}}\nabla_{\bm{\phi}}L^{VF}\left(\bm{\phi}\right),
	\label{50}
\end{equation}
where $\eta_{\bm{\phi}}$ is the learning rate of critic network $\bm{\phi}$.

For the actor network of UE $u$, which is employed to generate variables $\left\{l_u, \mathbf{b}_u\right\}$, the loss function is expressed as
\begin{equation}
	\begin{aligned}
	&L^{CLIP}_u\left(\bm{\psi}_u\right)
	=\mathbb{E}_t\left[\mathbf{min}(\mathbf{r}_{u,t}\hat{A}_t,\mathbf{clip}\left(\mathbf{r}_{u,t},1-\kappa,1+\kappa\right)\hat{A}_t)\right],
	\end{aligned}
	\label{51}
\end{equation}
where $\mathbf{min}\left(\cdot\right)$ represents the minimum value of the corresponding element in the vector, $\mathbf{clip}\left(\cdot\right)$ represents the clipping function that limits the value of the corresponding element in the estimated action ratio vector to a range $\left[1-\kappa,1+\kappa\right]$, and $\mathbf{r}_{u,t}$ is the estimated action ratio vector determined by the different network parameters $\bm{\psi}_u$ and $\bm{\psi}^{old}_u$, namely
\begin{equation}
	\mathbf{r}_{u,t}=\frac{O_t\left(s_u\left(t\right)|\bm{\psi}_u\right)}{O_t\left(s_u\left(t\right)|\bm{\psi}_u^{old}\right)}.
	\label{52}
\end{equation}
The update of UE $u$'s actor network is expressed as 
\begin{equation}
	\bm{\psi}_u\leftarrow\bm{\psi}_u+\eta_{\bm{\psi}}\nabla_{\bm{\psi}}L^{CLIP}_u\left(\bm{\psi}_u\right),
	\label{53}
\end{equation}
where $\eta_{\bm{\psi}}$ is the learning rate of actor networks.

The proposed MAIS algorithm for optimizing the long-term model splitting and AP clustering strategies within the UCSFL framework in the user-centric CF-MIMO network, is summarized in Algorithm \ref{algo2}.
Since the expected transmission rate can be approximated by separately averaging the useful and interference signals \cite{9014542}, Algorithm \ref{algo1} can be employed to compute an approximation of the expected uplink and downlink rate \cite{11126942}.
The impact of this approximation on the performance of Algorithm \ref{algo2} is negligible.
 
\begin{algorithm}[hbtp]
	\caption{~MAIS for Model Splitting and AP Clustering}
	\begin{algorithmic}[1]
		\State \textbf{Init:} Actor networks $\left\{\bm{\psi}_u,\bm{\psi}_u^{old}|u\in\mathcal{U}\right\}$, critic networks 
		\Statex\qquad\ $\left\{\bm{\phi},\bm{\phi}^{old}\right\}$, and initial states $\left\{\mathbf{s}_u\left(0\right)|u\in\mathcal{U}\right\}$ 
		\Repeat
		\For{$t\leq\mathbb{T}$}
		\State Use old actor networks to output actions $\mathcal{\hat{A}}(t)$
		\State Get next states $\mathcal{S}(t+1)$ and reward $U_t$
		\State Store the tuple $\left(\mathcal{S}(t),\mathcal{\hat{A}}(t),U_t,\mathcal{S}(t+1)\right)$ 
		\EndFor
		\State Calculate a series of TD errors by using Eq. \eqref{47}
		\State Calculate a series of GAEs by using Eq. \eqref{48}
		\State Calculate the critic network's loss by using Eq. \eqref{49}
		\State Update the critic network by using Eq. \eqref{50}
		\State Calculate action ratio vectors by using Eq. \eqref{52}
		\State Calculate actor networks' losses by using Eq. \eqref{51}
		\State Update actor networks by using Eq. \eqref{53}
		\State Clear the replay buffer
		\Until convergence
		\State \textbf{Output} $\left\{l_u,\mathbf{b}_u|u\in\mathcal{U}\right\}$
	\end{algorithmic}
	\label{algo2}
\end{algorithm}

\section{Simulation and Discussion}
\label{Simulation}

This section presents simulation results for the proposed UCSFL framework in the user-centric CF-MIMO network, together with comparisons against representative baseline schemes.
In addition, an MNIST-based handwritten digit classification example is provided to demonstrate the effectiveness of adaptive model splitting and AP clustering in improving training performance.

\subsection{Parameters and Setup}

We consider a user-centric CF-MIMO network area having a radius $r=200 ~\rm{m}$ with $M=10$ APs and $U=4$ uniformly distributed UEs. UEs. To avoid boundary effects, wrap-around is adopted. The channel vector $\mathbf{h}_{m,u}$ is modeled as 
\begin{equation}
	\mathbf{h}_{m,u}=\sum\limits_{l}{\sqrt{L_{m,u}}e^{-j2\pi{f}\tau_{m,u,l}}\mathbf{g}_{m,u,l}},
	\tag{49}
	\label{eq49}
\end{equation}
where $\mathbf{g}_{m,u,l}\sim\mathcal{CN}\left(0,\mathbf{I}_{L}\right)$ represents the small-scale fading coefficients, $f$ represents the frequency shift, and $\tau_{m,u,l}$ represents the latency of path $l$.
The large-scale fading coefficients $L_{m,u}$ can be modeled as \cite{bjornson2019making}
\begin{equation}
	L_{m,u}=10^{\frac{PL_{mu}^{\rm{d}}}{10}}10^{\frac{F_{mu}}{10}},
	\tag{50}
	\label{eq50}
\end{equation}
where $10^{\left(PL_{mu}^{\rm{d}}/10\right)}$ represents the pass loss, and $10^{\left(F_{mu}/10\right)}$ represents the shadowing effect with $F_{mu}\in\mathcal{N}\left(0,4^2\right)$ (in dB).
Here, $PL_{mu}^{\rm{d}}$ (in dB) is given by 
\begin{equation}
	PL_{mu}^{\rm{d}}=-112.4-38\log_{10}{\left(\frac{d_{mu}}{1~\rm{m}}\right)},
	\tag{51}
	\label{eq51}
\end{equation}
where $d_{mu}$ is the distance between AP $m$ and UE $u$.

Furthermore, we employ VGG16 as the neural network model to implement the UCSFL framework, which contains 5 convolutional blocks and 1 fully connected block. 
Thus, the total number of splitting points $L$ is set to 6.
The model parameters corresponding to different splitting points are summarized in Table \ref{T1}, while other simulation parameters are listed in Table \ref{T2}.

\begin{table}[tp] % htbp代表表格浮动位置
	\centering
	\caption{VGG16 PARAMETERS}
	\begin{tabular}{ccccccc}
		\toprule
		\text{Parameter} & $l_u=1$ & $l_u=2$ & $l_u=3$ & $l_u=4$ & $l_u=5$ & $l_u=6$ \\
		\midrule
		$D_u$\text{~(Mb)} & 0.8 & 0.4 & 0.2 & 0.1 & 0.025 & 0.001 \\
		$D_{l_u}$\text{~(Mb)} & 0.039 & 0.261 & 1.741 & 7.641 & 14.721 & 138.361 \\
		$x_{l_u}$\text{~(G)} & 3.87 & 9.42 & 18.67 & 27.92 & 30.69 & 30.94 \\
		\bottomrule
	\end{tabular}
	\label{T1}
\end{table}

\begin{table}[tp] % htbp代表表格浮动位置
	\centering
	\caption{SIMULATION PARAMETERS}
	\begin{tabular}{cccc}
		\toprule
		\text{Parameter} & \text{~Value} & \text{Parameter} & \text{~Value} \\
		\midrule
		$f^{dpu}$ & 5~\rm{GHz} & $c$ & 1~\rm{cycles} \\
		$P^{ul}$ & 100~\rm{mW} & $P^{dl}$ & 300~\rm{mW} \\
		$\rho^{\theta}$ & $1\times10^{3}$ & $\rho^{\pi}$ & $1\times10^{-3}$ \\
		$\gamma$ & 0.99 & $\lambda$ & 0.9 \\
		$\mathbb{T}$ & 4 & $\kappa$ & 0.1 \\
		$\eta_{\bm{\phi}}$ & $1\times10^{-5}$ & $\eta_{\bm{\psi}}$ & $5\times10^{-5}$ \\
		\bottomrule
	\end{tabular}
	\label{T2}
\end{table}

To further evaluate the effectiveness of the proposed UCSFL, we consider the following baseline schemes:
\begin{itemize}
	\item Baseline 1 (BL1) : A simplified UCSFL setting that adopts the same model splitting strategy as UCSFL, while associating each UE with all APs, i.e., $C_u = M, \forall u\in\mathcal{U}$.
	\item Baseline 2 (BL2) : A simplified UCSFL framework that adopts the same user association strategy as UCSFL, while forcing all UEs to split at the first layer, i.e., $l_u = 1, \forall u\in\mathcal{U}$. 
	\item Baseline 3 (BL3) : A conventional FL framework in which each UE is associated with all APs and no model splitting is performed, i.e., $C_u = M, \forall u\in\mathcal{U}$, and $l_u = L, \forall u\in\mathcal{U}$. 
\end{itemize}

In addition, we adopt the VGG16 model within the UCSFL framework for MNIST-based handwritten digit classification, simulating a scenario with 4 UEs. The training configuration includes a cross-entropy loss function, a batch size of 512, and a learning rate of 0.1.

\subsection{Results and Discussion}

\begin{figure}[tp]
	\centering
	\includegraphics[scale=0.5]{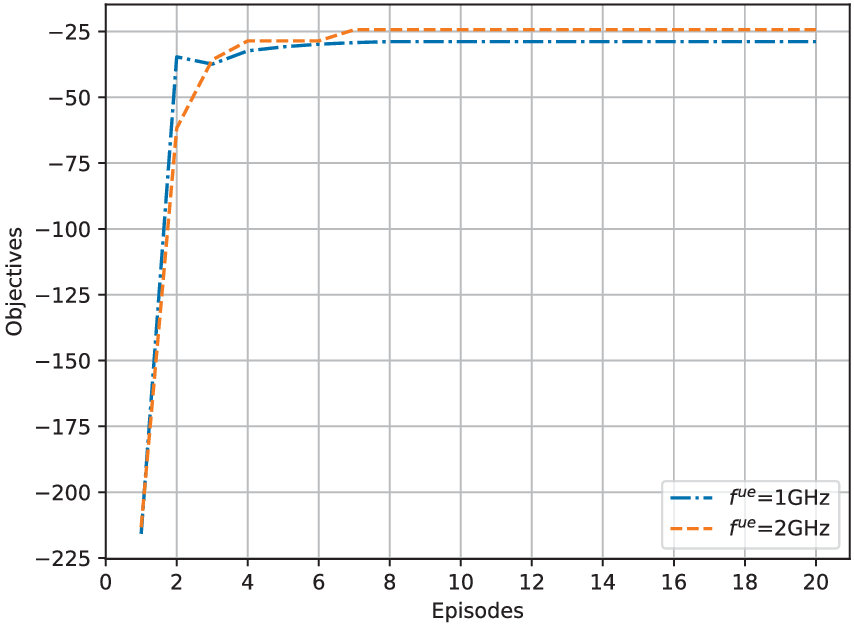}
	\caption{Convergence behavior of MAIS.}
	\vspace{12pt}
	\label{5_1}
\end{figure}

First, we evaluate the convergence behavior of the proposed algorithm. It can be seen from Fig. \ref{5_1} that MAIS converges within fewer than 10 episodes. Notably, as the UE computing frequency increases, the convergence speed slows down while the objective value increases, and the curve becomes smoother. This occurs because higher computational capacity reduces training latency, which in turn elevates the objective reward and stabilizes the learning process under a fixed learning rate. 

\begin{figure}[tp]
	\centering
	\includegraphics[scale=0.5]{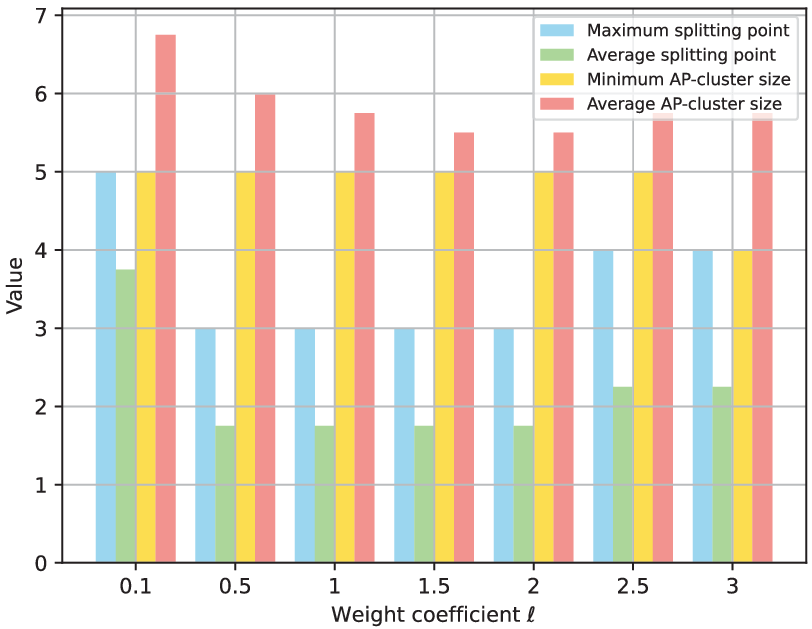}
	\caption{Splitting and clustering strategies versus weight coefficient. Here, we set $f^{ue}=1~\rm{GHz}$ and $w=15~\rm{kHz}$.}
	\vspace{12pt}
	\label{5_2}
\end{figure}

Fig. \ref{5_2} illustrates how the model splitting and AP clustering strategies vary with the accuracy-weight coefficient $\ell$ in the latency-to-accuracy ratio.
For AP clustering, as weight coefficient $\ell$ increases, the minimum AP-cluster size among UEs remains largely unchanged, whereas the average AP-cluster size first decreases and then increases over the range $[0.1, 3]$, consistent with the trend of individual UEs’ AP-cluster sizes. 
Since a larger AP-cluster size is positively correlated with training accuracy, the contribution of clustering to the objective is minimized when $\ell\in\{1.5, 2\}$ and becomes more pronounced as $\ell$ deviates from this range, especially when $\ell = 0.1$. 
This is because the AP-cluster size is jointly magnified by the denominator and an exponent significantly smaller than one, thereby increasing its relative influence on the objective. 
For the splitting strategy, both the maximum and average splitting points exhibit a similar trend of first decreasing and then increasing, achieving their minimum values when $\ell \in \{0.5, 1, 1.5, 2\}$. 
This suggests that the splitting strategy is less sensitive to $\ell$ than the clustering strategy. 
As shown in Table \ref{T1}, increasing the average splitting point reduces the uplink activation data volume, which helps alleviate the latency degradation caused by reduced uplink rate under stronger inter-user interference resulting from larger cluster sizes.

Fig. \ref{5_3} shows how the latency of the worst-performing UE varies with the user computing frequency $f^{ue}$.
As $f^{ue}$ increases from 1 GHz to 3 GHz, the training latency gradually decreases from 144 s to 118 s, yet the reduction rate gradually diminishes. 
This indicates that training latency becomes progressively less sensitive to increases in user computing frequency $f^{ue}$, since transmission latency dominates the overall delay.
Comparing the red and yellow bars reveals that the uplink latency is substantially greater than the downlink latency, which validates the analysis in Fig. \ref{5_2}. 
Meanwhile, the yellow bars drop abruptly at 1.5 GHz and then remain nearly constant, suggesting a change in user association between 1 GHz and 1.5 GHz that leads to a significant variation in downlink transmission rate. 
Similarly, the red bars decrease sharply at 2 GHz and then stabilize, indicating another shift in user association between 1.5 GHz and 2 GHz, resulting in a notable change in uplink transmission rate. 
In general, the AP clustering strategy undergoes more frequent adjustments at lower user computing frequencies, whereas it remains largely unaffected by frequency changes at higher levels.
In addition, the comparison between the blue and gray bars shows that the gray bars decline gradually while the blue bars remain unchanged, implying that the splitting strategy is kept fixed, which indirectly confirms the variation in AP clustering.

\begin{figure}[tp]
	\centering
	\includegraphics[scale=0.5]{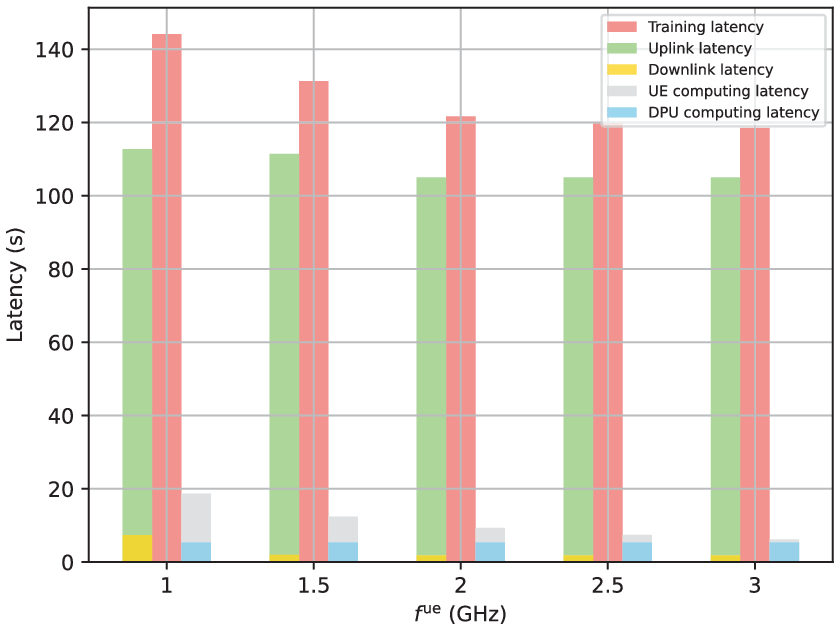}
	\caption{Latency of the worst-performing UE versus user computing frequency. Here, we set $w=15~\rm{kHz}$ and $\ell=1$.}
	\vspace{12pt}
	\label{5_3}
\end{figure}

\begin{figure}[tp]
	\centering
	\includegraphics[scale=0.5]{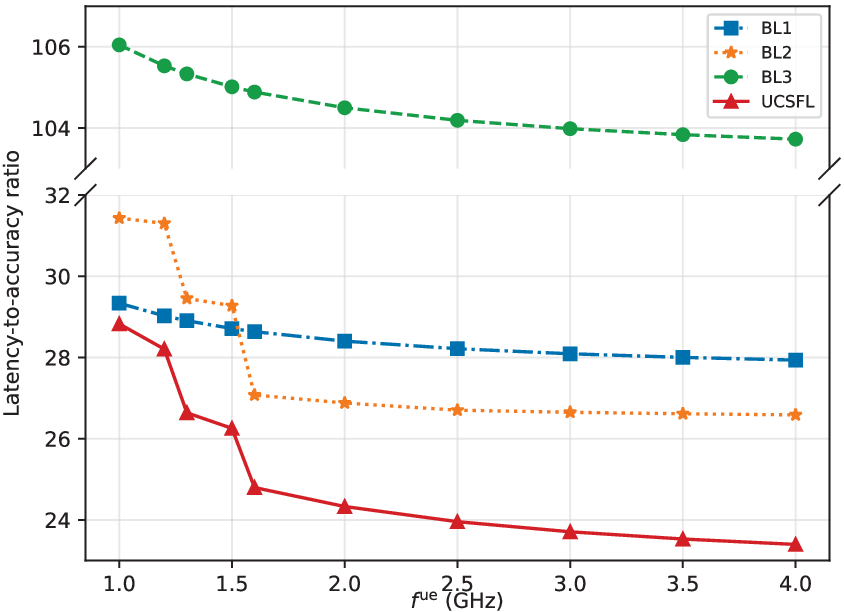}
	\caption{Latency-to-accuracy ratio versus user computing frequency. Here, we set $w=15~\rm{kHz}$ and $\ell=1$.}
	\vspace{12pt}
	\label{5_4}
\end{figure}

Fig. \ref{5_4} illustrates how the maximum latency-to-accuracy ratio varies with user computing frequency $f^{ue}$.
Since the performance gap between BL1 and BL3 remains largely constant across different $f^{ue}$, the model splitting strategy in the proposed UCSFL remains unchanged under current communication conditions, i.e, $\mathrm{w}=15$ kHz. 
This demonstrates that an appropriate splitting strategy can significantly reduce training latency. 
Furthermore, comparing BL2 with the UCSFL reveals that the MASI algorithm adaptively adjusts the splitting strategy to enhance system performance.
Additionally, the latency-to-accuracy ratio of both BL2 and UCSFL drops rapidly at $f^{ue}=1.3$ GHz and $f^{ue}=1.6$ GHz, indicating that the proposed algorithm adaptively optimizes the AP clustering strategy, which also corroborates the analysis related to Fig. \ref{5_3}. 
We also observe that the performance gap between BL1 and UCSFL widens rapidly with increasing $f^{ue}$ at lower values before gradually stabilizing. 
This occurs because the training latency reduction gain outweighs the effect of smaller AP-cluster size from UCSFL's adaptive clustering.
Moreover, the intersection of BL1 and BL2 within $f^{ue} \in (1.5, 1.6)$ GHz suggests that a more suitable model splitting strategy yields greater benefits than AP clustering optimization at relatively low computing frequencies. 
However, as $f^{ue}$ increases, their comparative advantages reverse.

\begin{figure}[tp]
	\centering
	\includegraphics[scale=0.5]{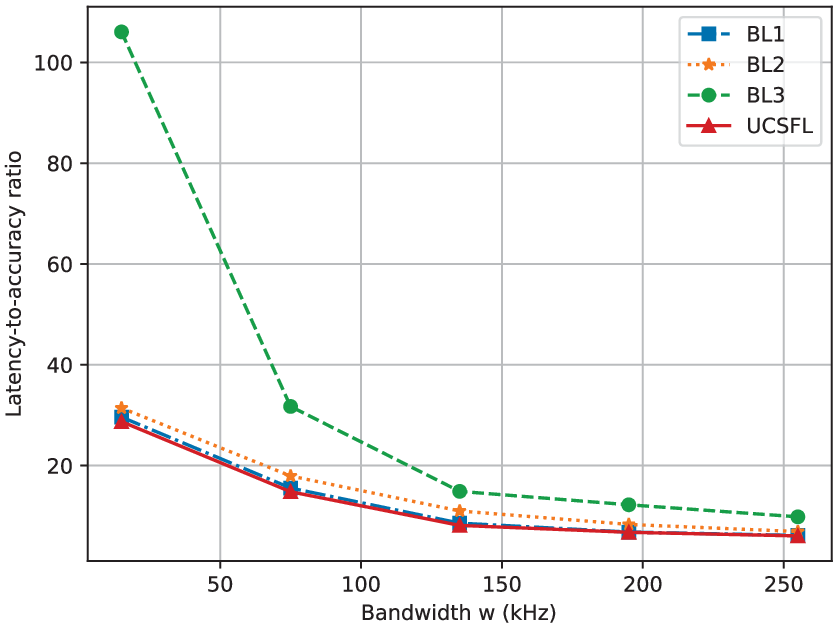}
	\caption{Latency-to-accuracy ratio versus bandwidth. Here, we set $f^{ue}=1~\rm{GHz}$ and $\ell=1$.}
	\vspace{12pt}
	\label{5_5}
\end{figure}

\begin{figure}[tp]
	\centering
	\includegraphics[scale=0.5]{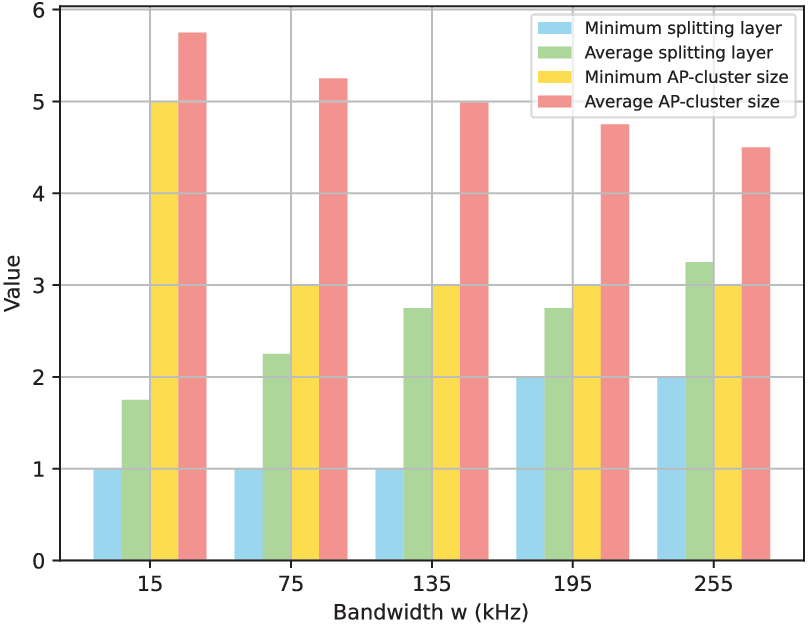}
	\caption{Splitting and clustering strategies versus bandwidth. Here, we set $f^{ue}=1~\rm{GHz}$ and $\ell=1$.}
	\vspace{12pt}
	\label{5_6}
\end{figure}

Fig. \ref{5_5} depicts the maximum latency-to-accuracy ratio across different schemes as bandwidth $\rm{w}$ varies.
The ratio decreases at a diminishing rate with increasing $\rm{w}$, owing to the exponential decay in transmission latency caused by higher transmission rates. 
As the transmission rate rises substantially, the latency gap between BL3 and the splitting-enabled schemes (BL1, BL2, and UCSFL) narrows, leading to a stabilization of the difference in their latency-to-accuracy ratios.
The results also show that BL1 consistently outperforms BL2. 
Although BL2 achieves a higher transmission rate through an appropriate clustering strategy, its latency advantage diminishes with bandwidth $\rm{w}$.
In contrast, BL1's adaptive splitting strategy reduces uplink data volume, which provides additional latency reduction and sustains its advantage.
Meanwhile, by employing larger AP-cluster sizes, BL1 boosts training accuracy with minimal latency overhead, thus ensuring its better performance.
Similarly, compared to BL1, UCSFL's adaptive clustering improves the transmission rate and reduces latency. 
However, this gain diminishes with increasing bandwidth $\rm{w}$ and is accompanied by smaller AP-cluster sizes, leading to a gradually shrinking performance gap between the two schemes.

In Fig. \ref{5_6}, we show how the splitting and clustering strategies change as the bandwidth $\rm{w}$ varies.
Regarding the model splitting strategy, as the bandwidth $\rm{w}$ increases, the minimum and average splitting layers among UEs gradually increase. 
This indicates that all UEs tend to split their models at higher layers, leading to increased computing latency, larger sub-model data volume, but reduced activation data volume, as outlined in Table \ref{T1}.
Given that the downlink rate substantially exceeds the uplink rate, a larger split layer can holistically reduce the transmission latency. 
Meanwhile, the accompanying rise in computing latency is undoubtedly smaller than the reduction in transmission latency, especially at higher transmission rates.
For the AP clustering strategy, the average AP-cluster size exhibits a trend of decrease, while the minimum AP-cluster size first decreases and then stabilizes, which keeps the worst-case latency-to-accuracy ratio within an acceptable range.
Although a reduction in the average AP-cluster size directly lowers the training accuracy, it helps mitigate inter-user interference, thereby improving transmission rates and reducing training latency.

\begin{figure}[tp]
	\centering
	\includegraphics[scale=0.5]{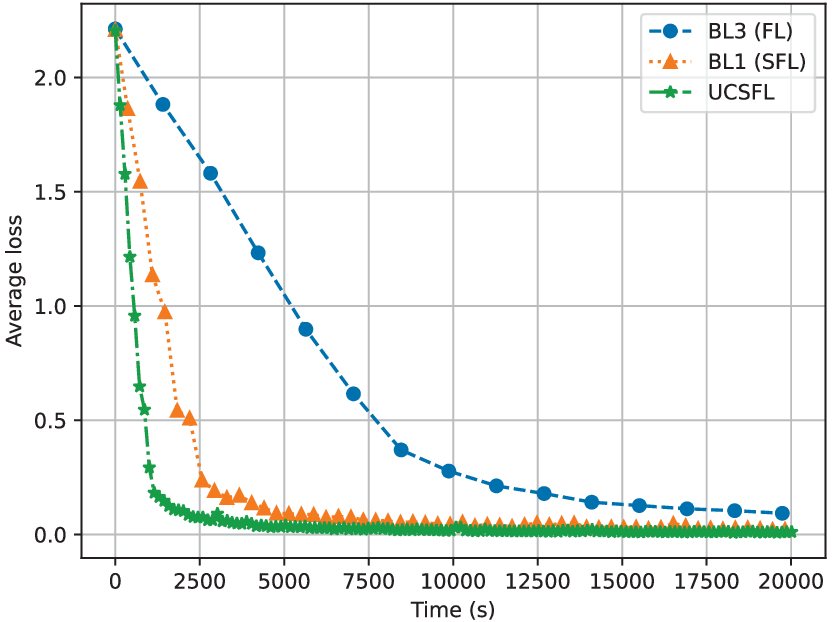}
	\caption{Average loss value versus training time. Here, we set $f^{ue}=1~\rm{GHz}$, $w=15~\rm{kHz}$, and $\ell=1$.}
	\vspace{12pt}
	\label{5_7}
\end{figure}

In Fig. \ref{5_7}, we show how the average loss value changes as the training time varies in the examples of implementing UCSFL for MNIST-based handwritten digit classification, where each point on the curve represents a single training iteration. Here, BL3 represents a federated averaging (FedAvg) scheme where CF-MIMO network is only used for transmission, and UEs' updated models are aggregated at the CPU.
It is evident that UCSFL exhibits a faster convergence speed than both BL1 and BL3.
Specifically, UCSFL requires only one-tenth the time of BL3 and one-quarter the time of BL2 to achieve the same average loss level.
This demonstrates that our proposed UCSFL scheme significantly reduces the training latency through its adaptive model splitting and AP clustering strategies, thus accelerating the model's convergence.

\begin{figure}[tp]
	\centering
	\includegraphics[scale=0.5]{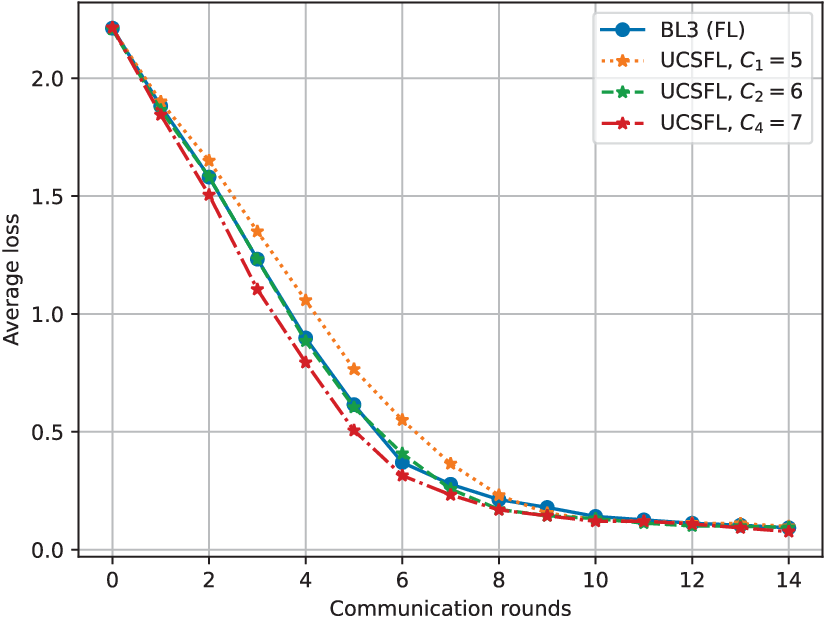}
	\caption{Average loss value versus communication rounds. Here, we set $f^{ue}=1~\rm{GHz}$, $w=15~\rm{kHz}$, and $\ell=1$.}
	\vspace{12pt}
	\label{5_8}
\end{figure}

Fig. \ref{5_8} compares the average loss versus communication rounds for UEs across different AP-cluster sizes in UCSFL and under the baseline BL1 (FL) scheme.
Here, $C_1=5$ corresponds to the first UE in the UCSFL scheme associated with 5 APs, while $C_2=6$ and $C_4=7$ denote the second and fourth UEs associated with 6 and 7 APs, respectively.
A comparison of the orange, green, and red curves reveals that UEs with larger AP-cluster sizes in UCSFL achieve faster convergence, which validates the corollary of Theorem \ref{th1} that a larger $C_u$ yields higher model accuracy. 
However, after more than 10 communication rounds, all curves converge to a similar loss value.
This phenomenon arises from the distributed aggregation of UE model parameters through their associated APs.
Global model information can be progressively collected over multiple communication rounds, leading to the gradual convergence of loss values across all UEs.
Meanwhile, the models of BL3 and different UEs in UCSFL	converge to a similar loss value, indicating that the proposed UCSFL achieves comparable performance to FedAvg scheme.

\section{Conclusion}
\label{Conclusion}

In this paper, we proposed a user-centric split federated learning (UCSFL) framework that enables the implementation of SFL over user-centric CF-MIMO network, in which AP clustering and a two-level model aggregation are performed. The convergence analysis demonstrates that the AP-cluster size is a key factor influencing model training accuracy. Motivated by this insight, we defined the latency-to-accuracy ratio as a unified metric to capture and optimize the trade-off between training latency and accuracy.
To solve the resulting joint optimization problem, we developed an NBCD algorithm for the short-term design of uplink power control and downlink beamforming, and a MAIS algorithm for the long-term optimization of model splitting and AP clustering.
Simulation results under representative parameter settings demonstrate that the proposed UCSFL significantly reduces the latency-to-accuracy ratio and accelerates VGG16 training convergence compared with baseline schemes.

\ifCLASSOPTIONcaptionsoff
\newpage
\fi

\bibliographystyle{IEEEtran}
\bibliography{ReferencesNewAbbr}

\appendix
\section{A. Proof}
\subsection{Proof of Theorem \ref{th1}}
\label{proof A}
	
	Firstly, we define the gradient for UE $u$'s model update as
	\begin{equation}
		g_{u,t}=
		\frac{1}{C_u}\sum\limits_{m}{b_{m,u}\frac{1}{C_m}\sum\limits_{v}{b_{m,v}\nabla{}F_{m,v,t}\left(\mathbf{W}_{m,v,t}~|~\zeta_{v,t}\right)}},
		\label{54}
	\end{equation}
	while the expected gradient for UE $u$'s model update is expressed as
	\begin{equation}
		\overline{g}_{u,t}=
		\frac{1}{C_u}\sum\limits_{m}{b_{m,u}\frac{1}{C_m}\sum\limits_{v}{b_{m,v}\nabla{}F_{m,v,t}\left(\mathbf{W}_{m,v,t}\right)}}.
		\label{55}
	\end{equation}
	Therefore, the equivalent update of UE $u$'s model is expressed as
	\begin{equation}
		\overline{\mathbf{W}}_{u,t+1}=\overline{\mathbf{W}}_{u,t}-\eta_{t}g_{u,t}
		\label{56}
	\end{equation}
	
	On this basis, we analyze the expected difference between UE’s model parameter and optimal parameter as follows:
	\begin{equation}
		\begin{aligned}
			&\mathbb{E}\vert\vert {\overline{\mathbf{W}}_{u,t+1}-\mathbf{W}^{\ast}} \vert\vert^2\\
			&=\mathbb{E}\vert\vert {\overline{\mathbf{W}}_{u,t}-\eta_{t}g_{u,t}-\mathbf{W}^{\ast}+\eta_{t}\overline{g}_{u,t}-\eta_{t}\overline{g}_{u,t}} \vert\vert^2\\
			&=\mathbb{E}\vert\vert{\overline{\mathbf{W}}_{u,t}-\eta_{t}\overline{g}_{u,t}-\mathbf{W}^{\ast}}\vert\vert^2
			+\eta_{t}^2\mathbb{E}\vert\vert{g_{u,t}-\overline{g}_{u,t}}\vert\vert^2\\
			&\quad+2\eta_{t}\mathbb{E}\langle{\overline{\mathbf{W}}_{u,t}-\eta_{t}g_{u,t}-\mathbf{W}^{\ast},~g_{u,t}-\overline{g}_{u,t}}\rangle.
		\end{aligned}
		\label{57}
	\end{equation}
	In order to present the analysis process more clearly, we analyze each item separately.
	The first term in Eq. \eqref{57} is expanded as follows:
	\begin{equation}
		\begin{aligned}
			&\mathbb{E}\vert\vert{\overline{\mathbf{W}}_{u,t}-\eta_{t}\overline{g}_{u,t}-\mathbf{W}^{\ast}}\vert\vert^2\\
			&=\mathbb{E}\vert\vert{\overline{\mathbf{W}}_{u,t}-\mathbf{W}^{\ast}}\vert\vert^2
			+\eta_{t}^2\mathbb{E}\vert\vert{\overline{g}_{u,t}}\vert\vert^2\\
			&\quad-2\eta_{t}\mathbb{E}\langle{\overline{\mathbf{W}}_{u,t}-\mathbf{W}^{\ast},\overline{g}_{u,t}}\rangle.
		\end{aligned}
		\label{58}
	\end{equation}
	In order to simplify the following expressions, we represent $\frac{1}{C_u}\sum\limits_{m}{b_{m,u}\frac{1}{C_m}\sum\limits_{v}{b_{m,v}}(\cdot)}$ as $\mathbb{M}(\cdot)$.
	Thus, an upper bound for the third term in Eq. \eqref{58} is given by:
	\begin{equation}
		\begin{aligned}
			&-2\eta_{t}\mathbb{E}\langle{\overline{\mathbf{W}}_{u,t}-\mathbf{W}^{\ast},\overline{g}_{u,t}}\rangle\\
			&=-2\eta_{t}\mathbb{M}(\mathbb{E}\langle{\overline{\mathbf{W}}_{u,t}-\mathbf{W}^{\ast},\nabla{}F_{m,v,t}}\rangle)\\
			&=-2\eta_{t}\mathbb{M}(
					\mathbb{E}\langle{\overline{\mathbf{W}}_{u,t}-\mathbf{W}_{m,v,t},\nabla{}F_{m,v,t}}\rangle)\\
			&\quad-2\eta_{t}\mathbb{M}(
					\mathbb{E}\langle{\mathbf{W}}_{m,v,t}-\mathbf{W}^{\ast},\nabla{}F_{m,v,t}\rangle)\\
			&\leq\eta_{t}\mathbb{M}[\mathbb{E}(\frac{1}{\eta{}_{t}}\left\|{\overline{\mathbf{W}}_{u,t}-\mathbf{W}_{m,v,t}}\right\|^2+\eta{}_t\left\|{\nabla{}F_{m,v,t}}\right\|^2)]\\
			&\quad-2\eta_{t}\mathbb{M}(
					\mathbb{E}\langle{\mathbf{W}}_{m,v,t}-\mathbf{W}^{\ast},\nabla{}F_{m,v,t}\rangle) \\
			&\leq\eta_{t}\mathbb{M}[\mathbb{E}(\frac{1}{\eta{}_{t}}\left\|{\overline{\mathbf{W}}_{u,t}-\mathbf{W}_{m,v,t}}\right\|^2+\eta{}_t\left\|{\nabla{}F_{m,v,t}}\right\|^2)]\\
			&\quad-2\eta_{t}\mathbb{M}[
					\mathbb{E}(F_{m,v,t}-F^{\ast}+\frac{\mu}{2}\left\|\mathbf{W}_{m,v,t}-\mathbf{W}^{\ast}\right\|^2)],
		\end{aligned}
		\label{59}
	\end{equation}
	where the second to last inequality is due to the AM-GM inequality, and the last inequality is due to the $\mu$-strongly convex (\textbf{Assumption 2}).
	Similarly, an upper bound for the second term in Eq. \eqref{58} is given by:
	\begin{equation}
		\begin{aligned}
			&\eta_{t}^2\mathbb{E}\left\|{\overline{g}_{u,t}}\right\|^2\\
			&=\eta_{t}^2\mathbb{E}\left\|\mathbb{M}(\nabla{}F_{m,v,t})\right\|^2\\
			&\leq{\eta_{t}^2LZ^2}\frac{C_u+1}{2C_u\rm{M}}\sum\limits_{m}{\frac{C_m+1}{2C_m}}\\
			&\leq{\eta_{t}^2LZ^2}\frac{C_u+1}{2C_u},
		\end{aligned}
		\label{60}
	\end{equation}
	where the second to last inequality is due to the proof of \textit{Lemma} 2 in \cite{9923620} and \textbf{Assumption 4}, the last inequality is due to the fact that $C_m\geq{}1,\forall{}m\in\mathcal{M}$, which means $\sum\limits_{m}{\frac{C_m+1}{2C_m}}\leq{}\rm{M}$.
	By the $\beta$-smoothness (\textbf{Assumption 1}), we have
	\begin{equation}
		\left\|{\nabla{}F_{m,v,t}}\right\|^2\leq2\beta\left(F_{m,v,t}-F_{m,v,t}^{\ast}\right).
		\label{61}
	\end{equation}
	Thus, by defining $\gamma_t=2\eta_t\left(1-\beta\eta_t\right)$, we can derive the expression based on Eq. \eqref{59} as follows: 
	\begin{equation}
		\begin{aligned}
			&\mathbb{M}(
					2\beta\eta_t^2\beta\mathbb{E}\left\{F_{m,v,t}-F_{m,v,t}^{\ast}\right\})-\mathbb{M}(
					2\eta_t\mathbb{E}\left\{F_{m,u,v}-F^{\ast}\right\})\\
			&=-\gamma_t\mathbb{M}(
					\mathbb{E}\left\{F_{m,v,t}-F_{m,v,t}^{\ast}\right\})+2\eta_t\mathbb{M}(
					\mathbb{E}\left\{F^{\ast}-F_{m,v,t}^{\ast}\right\})\\
			&=-\gamma_t\mathbb{M}(
					\mathbb{E}\left\{F_{m,v,t}-F^{\ast}\right\})+\left(2\eta_t-\gamma_t\right)\mathbb{M}(
					\mathbb{E}\left\{F^{\ast}-F_{m,v,t}^{\ast}\right\})\\
			&=-\gamma_t\mathbb{M}(
					\mathbb{E}\left\{F_{m,v,t}-F^{\ast}\right\})+2\beta\eta_t^2\Gamma\\
			&\leq-\frac{\gamma_t\mu}{2}\mathbb{M}(
					\mathbb{E}\|\mathbf{W}_{m,v,t}-\mathbf{W}^{\ast}\|^2)+2\beta\eta_t^2\Gamma,
		\end{aligned}
		\label{62}
	\end{equation}
	where $\Gamma=\mathbb{M}(
	\mathbb{E}\left\{F^{\ast}-F_{m,v,t}^{\ast}\right\})$, and the last inequality is due to the $\mu$-strongly convex (\textbf{Assumption 2}). Therefore, an upper bound for Eq. \eqref{58} is given by:
	\begin{equation}
		\begin{aligned}
			&\mathbb{E}\vert\vert{\overline{\mathbf{W}}_{u,t}-\eta_{t}\overline{g}_{u,t}-\mathbf{W}^{\ast}}\vert\vert^2\\
			&\leq\mathbb{E}\vert\vert{\overline{\mathbf{W}}_{u,t}-\mathbf{W}^{\ast}}\vert\vert^2+\mathbb{M}(
					\mathbb{E}\|\overline{\mathbf{W}}_{u,t}-\mathbf{W}_{m,v,t}\|^2)\\
			&\quad-\frac{\left(\gamma_t+2\eta_t\right)\mu}{2}\mathbb{M}(
					\mathbb{E}\|\mathbf{W}_{m,v,t}-\mathbf{W}^{\ast}\|^2)\\
			&\quad+2\beta\eta_t^2\Gamma+\eta_t^2LZ^2\frac{C_u+1}{2C_u}.
		\end{aligned}
		\label{63}
	\end{equation}
	Given that $\mathbb{M}(\mathbb{E}\|\mathbf{W}_{m,v,t}-\mathbf{W}^{\ast}\|^2)\geq\mathbb{E}\|\mathbb{M}(\mathbf{W}_{m,v,t}-\mathbf{W}^{\ast})\|^2$ and since $\eta_t\beta\leq1$, which means $1-2\mu\eta_t+\mu\beta\eta_t^2\leq1-\mu\eta_t$, we obtain a new upper bound for Eq. \eqref{58}:
	\begin{equation}
		\begin{aligned}
			&\mathbb{E}\vert\vert{\overline{\mathbf{W}}_{u,t}-\eta_{t}\overline{g}_{u,t}-\mathbf{W}^{\ast}}\vert\vert^2\\
			&\leq\left(1-\mu\eta_t\right)\mathbb{E}\vert\vert{\overline{\mathbf{W}}_{u,t}-\mathbf{W}^{\ast}}\vert\vert^2+\mathbb{M}(
					\mathbb{E}\|\overline{\mathbf{W}}_{u,t}-\mathbf{W}_{m,v,t}\|^2)\\
			&\quad+2\beta\eta_t^2\Gamma+\eta_t^2LZ^2\frac{C_u+1}{2C_u}.
		\end{aligned}
		\label{64}
	\end{equation}
	To further simplify the expression, we have
	\begin{equation}
		\begin{aligned}
			&\mathbb{E}\|\overline{\mathbf{W}}_{u,t}-\mathbf{W}_{m,v,t}\|^2\\
			&=\mathbb{E}\|\left(\overline{\mathbf{W}}_{u,t}-\mathbf{W}_{t_0}\right)+\left(\mathbf{W}_{t_0}-\mathbf{W}_{m,v,t}\right)\|^2\\
			&\leq2\mathbb{E}\|\overline{\mathbf{W}}_{u,t}-\mathbf{W}_{t_0}\|^2+2\mathbb{E}\|\mathbf{W}_{m,v,t}-\mathbf{W}_{t_0}\|^2\\
			&=2\mathbb{E}\|\overline{\mathbf{W}}_{u,t}-\mathbf{W}_{t_0}\|^2+2\mathbb{E}\|\eta_t\sum\limits_{\tau=t_0}^{t}\nabla{}F_{m,v,\tau}\|^2\\
			&\leq2\mathbb{E}\|\overline{\mathbf{W}}_{u,t}-\mathbf{W}_{t_0}\|^2+2\eta_t^2\left(T-1\right)^2LZ^2
		\end{aligned}
		\label{65}
	\end{equation}
	where $\mathbf{W}_{t_0}$ represents the initial model parameter, the first inequality is due to the Cauchy-Schwarz inequality, and the last inequality is due to the Jensen inequality with \textbf{Assumption 4}.
	For the first term in Eq. \eqref{65}, we can derive its upper bound expressed as 
	\begin{equation}
		\begin{aligned}
			&2\mathbb{E}\|\overline{\mathbf{W}}_{u,t}-\mathbf{W}_{t_0}\|^2\\
			&=2\mathbb{E}\|\mathbb{M}\left(\mathbf{W}_{m,v,t}-\mathbf{W}_{t_0}\right)\|^2\\
			&\leq\frac{C_u+1}{C_u\rm{M}}\sum\limits_{m}{\frac{C_m+1}{2C_m\rm{U}}\sum\limits_{v}{{\mathbb{E}\|\mathbf{W}_{m,v,t}-\mathbf{W}_{t_0}\|^2}}}\\
			&\leq\eta_t^2\left(T-1\right)^2LZ^2\frac{C_u+1}{C_u\rm{M}}\sum\limits_{m}{\frac{C_m+1}{2C_m}}\\
			&\leq\eta_t^2\left(T-1\right)^2LZ^2\frac{C_u+1}{C_u},
		\end{aligned}
		\label{66}
	\end{equation}
	where the first inequality is similar to the analysis process of Eq. \eqref{60}, the second inequality is due to the the Jensen inequality with \textbf{Assumption 4}, and the last inequality is due to the fact that $C_m\geq1,\forall{}m\in\mathcal{M}$.
	Therefore, a upper bound for the second term in Eq. \eqref{64} is given by:
	\begin{equation}
		\begin{aligned}
			&\mathbb{M}(
					\mathbb{E}\|\overline{\mathbf{W}}_{u,t}-\mathbf{W}_{m,v,t}\|^2)\\
			&\leq\eta_t^2\left(T-1\right)^2LZ^2\frac{3C_u+1}{C_u}.
		\end{aligned}
		\label{67}
	\end{equation} 
	Furthermore, an upper bound for the second term in Eq. \eqref{57} is derived by:
	\begin{equation}
		\begin{aligned}
			&\eta_t^2\mathbb{E}\vert\vert{g_{u,t}-\overline{g}_{u,t}}\vert\vert^2\\
			&=\eta_t^2\mathbb{E}\|\mathbb{M}(\nabla{}F_{m,v,t}(\zeta_{v,t})-\nabla{}F_{m,v,t})\|^2\\
			&\leq{\eta_t^2L\epsilon^2}\frac{C_u+1}{2C_u},
		\end{aligned}
		\label{68}
	\end{equation}
	where the inequality is similar to the analysis process of Eq. \eqref{60} with \textbf{Assumption 3}. Meanwhile, since $\mathbb{E}\{g_{u,t}\}=\overline{g}_{u,t}$, the third term in Eq. \eqref{57} consequently equals zero, i.e., $2\eta_{t}\mathbb{E}\langle{\overline{\mathbf{W}}_{u,t}-\eta_{t}g_{u,t}-\mathbf{W}^{\ast},~g_{u,t}-\overline{g}_{u,t}}\rangle=0$.
	
	In conclusion, an upper bound for the expected difference of Eq. \eqref{57} is expressed as
	\begin{equation}
		\begin{aligned}
			&\mathbb{E}\| {\overline{\mathbf{W}}_{u,t+1}-\mathbf{W}^{\ast}} \|^2\\
			&\leq\left(1-\mu\eta_t\right)\mathbb{E}\vert\vert{\overline{\mathbf{W}}_{u,t}-\mathbf{W}^{\ast}}\vert\vert^2+\eta_t^2R,
		\end{aligned}
		\label{69}
	\end{equation}
	where $R$ represents $\frac{3C_u+1}{C_u}\left(T-1\right)^2LZ^2+\frac{C_u+1}{2C_u}LZ^2+\frac{C_u+1}{2C_u}L\epsilon^2+2\beta\Gamma$.
	On this basis, we define $\Delta_{t+1}=\mathbb{E}\vert\vert {\overline{\mathbf{W}}_{u,t+1}-\mathbf{W}^{\ast}} \vert\vert^2$, and assume $\eta_t=\frac{\psi}{t+\iota}$ for some $\psi>\frac{1}{\mu}$ and $\iota>0$. Now, we prove $\Delta_t\leq\frac{v}{t+\iota}$ by induction, where $v=\max\left\{\frac{\psi^2R}{\psi\mu-1},(\iota+1)\Delta_1\right\}$.
	Obviously, the definition of $v$ ensures it holds for t = 1.
	For some $t$, 
	\begin{equation}
		\begin{aligned}
			\Delta_{t+1}&\leq\left(1-\mu\eta_t\right)\Delta_t+\eta_t^2R\\
			&\leq\left(1-\frac{\psi\mu}{t+\iota}\right)\frac{v}{t+\iota}+\frac{\psi^2R}{\left(t+\iota\right)^2}\\
			&=\frac{t+\iota-1}{\left(t+\iota\right)^2}v+\left[\frac{\psi^2R}{\left(t+\iota\right)^2}-\frac{\psi\mu-1}{\left(t+\iota\right)^2}v\right]\\
			&\leq\frac{v}{t+\iota+1}.
		\end{aligned}
		\label{70}
	\end{equation}
	Then by the $\beta$-smoothness of loss function $F(\cdot)$, we have
	\begin{equation}
		\mathbb{E}\left[F\left(\overline{\mathbf{W}}_{u,t}\right)-F^{\ast}\right]\leq\frac{\beta}{2}\Delta_t\leq\frac{\beta}{2}\frac{v}{t+\iota}
		\label{71}
	\end{equation}
	Specially, let $\psi=\frac{2}{\mu}$, $\alpha=\frac{\beta}{\mu}$, $\iota=\max\left\{8\alpha,T\right\}-1$, then
	\begin{equation}
		\begin{aligned}
			v&=\max\left\{{\frac{\psi^2R}{\psi\mu-1},\left(\iota+1\right)\Delta_1}\right\}\\
			&\leq\frac{\psi^2R}{\psi\mu-1}+\left(\iota+1\right)\Delta_1\\
			&\leq\frac{4R}{\mu^2}+\left(\iota+1\right)\Delta_1.
		\end{aligned}
		\label{72}
	\end{equation}
	Therefore, we have	
	\begin{equation}
		\begin{aligned}
			\mathbb{E}\left[F\left(\overline{\mathbf{W}}_{u,t}\right)\right]-F^{\ast}&\leq\frac{\beta}{2}\frac{v}{t+\iota}\\
			&\leq\frac{\alpha}{t+\iota}\left(\frac{2R}{\mu}+\frac{\mu\left(\iota+1\right)}{2}\Delta_1\right).
		\end{aligned}
		\label{73}
	\end{equation}

\end{document}